%% file: template.tex
\documentclass{article}

\usepackage{arxiv}

\usepackage[utf8]{inputenc} % allow utf-8 input
\usepackage[T1]{fontenc}    % use 8-bit T1 fonts
\usepackage{hyperref}       % hyperlinks
\usepackage{url}            % simple URL typesetting
\usepackage{booktabs}       % professional-quality tables
\usepackage{amsfonts}       % blackboard math symbols
\usepackage{nicefrac}       % compact symbols for 1/2, etc.
\usepackage{microtype}      % microtypography
\usepackage{graphicx}
\usepackage{natbib}
\usepackage{doi}
\usepackage{amsmath}
\usepackage{amsthm}
\usepackage{amssymb}
\usepackage{adjustbox}
\usepackage{verbatim}

\newcommand{\Var}{\mathrm{Var}}
\newcommand{\Cov}{\mathrm{Cov}}
\newcommand{\Cor}{\mathrm{Cor}}
\newcommand{\ind}[1]{{\mathbf{1}}{\left(#1\right)}}
\newcommand{\f}{f}
\newcommand{\bx}{\mathbf{x}}
\newcommand{\npdf}[1]{\phi\left({#1}\right)}
\newcommand{\x}{\mathbf{x}}
\newcommand{\beq}{\begin{equation}}
\newcommand{\eeq}{\end{equation}}
\newcommand{\vT}{T^{(v)}}
\newcommand{\vM}{M^{(v)}}
\newcommand{\M}{M}
\newcommand{\iid}{\overset{iid}{\sim}} 
\newcommand{\g}{g}
\newcommand{\Data}{\mathcal D}
\newcommand{\KL}{\operatorname{KL}}
\newcommand{\bX}{\mathbf{X}}
\newcommand{\Gam}{\operatorname{Gam}}
\newcommand{\eps}{\epsilon}
\newcommand{\ghxi}[2]{g(\x_{#2}, T_{#1}, \M_{#1})}
\newcommand{\vtau}{\tau^{(v)}}
\newcommand{\E}{\mathrm{E}}
\newcommand{\epsilont}{\widetilde \epsilon}
\def\eqd{\,{\buildrel d \over =}\,} 

\newcommand{\aaa}{1.0}
\newcommand{\aab}{0.8}
\newcommand{\aac}{-0.3}
\newcommand{\aad}{-0.2}

\newtheorem{thm}{Theorem}
\newtheorem{lemma}{Lemma}
\theoremstyle{remark}
\newtheorem{remark}{Remark}

\title{Density Regression with \\Bayesian Additive Regression Trees}

\author{\hspace{\aaa cm}Vittorio Orlandi\\
	\hspace{\aaa cm}Dept. of Statistical Science\\
    \hspace{\aaa cm}Duke University\\
	\hspace{\aaa cm}Durham, NC 27708 \\
	\hspace{\aaa cm}\texttt{vdo@duke.edu} \\
	\And
	\hspace{\aab cm}Jared Murray \\
	\hspace{\aab cm}Dept. of Information, \\\hspace{\aab cm}Risk, and Operations Management\\
	\hspace{\aab cm}University of Texas, Austin\\
	\hspace{\aab cm}Austin, TX 78712 \\
	\hspace{\aab cm}\texttt{jared.murray@mccombs.utexas.edu}
	\AND
	\hspace{\aac cm}Antonio Linero \\
	\hspace{\aac cm}Dept. of Statistics and Data Sciences \\
	\hspace{\aac cm}University of Texas, Austin\\
	\hspace{\aac cm}Austin, TX 78712 \\
	\hspace{\aac cm}\texttt{antonio.linero@austin.utexas.edu}
	\And
	\hspace{\aad cm}Alexander Volfovsky \\
	\hspace{\aad cm}Dept. of Statistical Science \\
	\hspace{\aad cm}Duke University\\
	\hspace{\aad cm}Durham, NC 27708 \\
	\hspace{\aad cm}\texttt{alexander.volfovsky@duke.edu}\\
}

\renewcommand{\shorttitle}{Density Regression with Bayesian Additive Regression Trees}

\hypersetup{
pdftitle={A template for the arxiv style},
pdfsubject={q-bio.NC, q-bio.QM},
pdfauthor={David S.~Hippocampus, Elias D.~Striatum},
pdfkeywords={First keyword, Second keyword, More},
}

\begin{document}
\maketitle

\begin{abstract}
	Flexibly modeling how an entire density changes with covariates is an important but challenging generalization of mean and quantile regression. While existing methods for density regression primarily consist of covariate-dependent discrete mixture models, we consider a continuous latent variable model in general covariate spaces, which we call DR-BART. The prior mapping the latent variable to the observed data is constructed via a novel application of Bayesian Additive Regression Trees (BART). We prove that the posterior induced by our model concentrates quickly around true generative functions that are sufficiently smooth. We also analyze the performance of DR-BART on a set of challenging simulated examples, where it outperforms various other methods for Bayesian density regression. Lastly, we apply DR-BART to two real datasets from educational testing and economics, to study student growth and predict returns to education. Our proposed sampler is efficient and allows one to take advantage of BART's flexibility in many applied settings where the entire distribution of the response is of primary interest. Furthermore, our scheme for splitting on latent variables within BART facilitates its future application to other classes of models that can be described via latent variables, such as those involving hierarchical or time series data. 
\end{abstract}

% keywords can be removed
\keywords{Bayesian Nonparametrics \and Conditional Density Estimation \and Posterior Concentration \and Latent Variables \and Heteroscedasticity}

\section{Introduction}
\label{sec:intro}
Data analysis frequently concerns itself with associating the change in a function of some response variable $y$ with a set of covariates $\bx$. Arguably the most common tool for this is mean regression, which focuses on the expectation $\mathrm{E}[y \mid \bx]$ and foregoes inference about other parts of the conditional density $p(y \mid \bx)$, a much more general quantity. This inflexibility has been recognized as problematic in many modern applications \citep[see, e.g.][]{dr_redshift, bdr_for_causal, prob_forecasting}. It is then natural to ask if other functionals of $p(y \mid \bx)$ are more appropriate, with two immediate candidates: quantile regression and density regression. 

Quantile regression models specific quantiles of the conditional response distribution. This can help address settings where some quantiles carry more probative value. For example, \cite{black_galdo_smith} note that the impact of a job training program on the upper quantiles of the distribution is considered to be more important by policy makers than that on the lower quantiles. While this approach is more general than mean regression, one problem with its application is that computing estimates of functionals of the quantiles is not always straightforward. Another problem is that estimated quantiles oftentimes do not obey the monotonicity constraint inherently satisfied by the true distribution.  While there are approaches for joint modeling of quantiles (e.g. \cite{tokdar_joint_quant_reg, rkhs_joint_quant_reg}) or post-hoc reordering of estimates (e.g. \cite{Chernozhukov2010}), the fundamental limitation with the approach is that individual quantiles are being targeted as proxies for features of the distribution as a whole. 

By modeling the entire probability distribution of the response, density regression methods perform a substantially harder task than mean or even quantile regression. In doing so, however, they are able to compute coherent point estimates and perform uncertainty quantification for arbitrary functionals of the distribution that may be of interest. Studies of income inequality, for example, typically take into account the impact of variables on the entire income distribution \citep[e.g.][]{income_inequality_swiss, income_inequality_fam}. To date, there is a large literature on density regression. We contribute to this body of work by providing a general, yet reasonably structured, formulation of the problem that enjoys important theoretical guarantees and yields an efficient sampler with strong empirical performance.

% This work proposes a model for density regression based off Bayesian Additive Regression Trees (BART) \citep{Chipman2010}. We place independent BART priors on components controlling the location and bandwidth of the density; the latter involves a novel extension of the BART prior to modeling variance functions. This makes our model very flexible, capable of accurately modeling a wide range of functions, which we show empirically in finite sample settings, as well as theoretically through our study of its posterior concentration. Our method also inherits BART's efficiency, yielding a well-mixing sampler at a relatively low computational cost.

In this work, we consider the following model for the conditional density:

\beq
p(y\mid \x, f, \sigma) = \int_0^1 
\frac{1}{\sigma(\x, u)}
%\phi\left(\frac{y-g(u)}{\sigma}\right) d\,u.\label{eq:con}
\npdf{\frac{y-f(\x, u)}{\sigma(\x, u)}}du,\label{eq:cmixfunc}
\eeq
which is parameterized by the latent variable $u$. This model generalizes the covariate dependent mixture model discussed in detail below. At the same time, it posits a fairly specific form by which $f$ and $\sigma$ interact to generate a conditional density. Various forms of this model have been considered in \cite{dunson_density}, \cite{Kundu2014}, and \cite{gp}, where Gaussian process priors were placed on $f(x, u)$ and inverse gamma priors on either $\sigma(x, u) = \sigma$ or $\sigma ^ 2(x, u) = \sigma ^ 2$. Instead, we propose placing BART-type priors on both $f$ and $\sigma$. While BART was designed for mean regression, we also consider a modification allowing it to model variance function to grant the model added flexibility in finite sample settings, and generalize both to accommodate latent variables. Asymptotically, we show that the posterior specified by our model concentrates around a true underlying density, provided its log is $\alpha$-H\"older for $0 < \alpha \leq 2$. The rate of concentration is removed from the minimax rate by a factor of $\alpha/(\alpha + 1)$ because BART generates piecewise constant functions; however, a near--minimax rate can be attained by using the SBART model of \cite{sbart} and restricting to a slightly smaller function class. In addition to yielding a more efficient sampler, we also show that this model outperforms its counterpart employing Gaussian process priors, as well as a variety of other methods for density regression in several empirical evaluations. 

We also apply our method to two real world datasets. In the first, we compute quantile growth targets for students in elementary and middle school mathematics classes using data provided in \cite{Betebenner2011SGP}. Mean regression approaches fail to identify interesting aspects of the data — such as that conditional distributions of test scores become more skewed over time — and quantile regression approaches can suffer from quantile crossing and a limited description of the uncertainty in their estimates. Our model addresses both these problems simultaneously. In the second, we study returns to education from US census microdata originally compiled by \cite{Angrist2006}. The returns are a nonlinear functional of quantiles of the wage distribution and therefore well suited for analysis by density regression methods, which fully capture uncertainty about their estimates. 

The paper proceeds as follows. The remainder of the introduction discusses past work on density regression models. Section \ref{sec:hetbart} gives a brief overview of BART and its application to modeling mean and variance functions. In Section \ref{sec:dr}, we motivate our use of BART for modeling components of the conditional density and state our full model for density regression, which we refer to as DR-BART. Section \ref{sec:theory} outlines theoretical results concerning our model, namely the rate at which the posterior contracts about a true density. Section \ref{sec:sims} compares DR-BART to other models for density regression on a variety of simulated datasets and Section \ref{sec:ex} applies DR-BART to two real world datasets from education and economics. Section \ref{sec:conclusion} concludes.
\subsection{Related Work}

Our proposal generalizes a common approach to density regression that uses covariate dependent mixture models. In these models, the conditional density is given by:
\beq
p(y\mid \x,\theta,\pi) = \sum_{h=1}^k \pi_h(\x)\mathcal{K}(y;\theta_h(\x)),\label{eq:cmix0}
\eeq
allowing for $k=\infty$, where $\mathcal{K}(\cdot, \cdot)$ is a positive definite kernel function and $\pi_h(\bx)$ are covariate-dependent mixture weights. We restrict attention to normal kernels: 
\beq
p(y\mid \x, \theta\equiv(\mu, \sigma),\pi) = \sum_{h=1}^k \pi_h(\x)
\phi_{\sigma_h(\x)}(y - \mu_h(\x))
,\label{eq:cmix00}
\eeq
where $\phi_{\sigma}(z) = (1/\sigma)\phi(z/\sigma)$ and $\phi$ denotes the standard normal pdf -- but the extension to other kernels is straightforward. In practice $\pi$ or $\theta$ may not vary with $\x$, or may vary in a limited way. (e.g., by taking $\mu_h(\x) = \x'\beta_h$ $\sigma_h(\x)= \sigma_h$). Our proposed model in \eqref{eq:cmixfunc} recovers that in \eqref{eq:cmix00} when $f$ and $\sigma$ are step functions with the same points of discontinuity (which can depend on the covariates).

Models of the form in \eqref{eq:cmix00} appear in the machine learning literature as ``mixtures of experts''  \citep{Jacobs1991,Jordan1994} where the initial focus was on using these models for flexible mean regression or classification.
\cite{Geweke2007} and \cite{Villani2009} study models of this form for semiparametric density regression, using finite $k$ and multinomial probit and logit regression models for $\pi(\x)$. While it is possible to get consistency properties for large classes of conditional densities (see e.g. \cite{Norets2010approximation,Pati2013,Norets2014} and the monograph \cite{Norets2014a}), practical experience in finite samples suggests that there can be value in allowing the kernel variance to depend on $\x$, as this can reduce the number of clusters required for an accurate approximation. \cite{Villani2009} provide simulated examples and discussion in the case of fixed $k$, and we will revisit this point in the context of the models introduced here (Section \ref{sec:sims}). 

Another approach proposes to leverage the joint model for $(y, \bx)$ as a convenient device for inducing a particular conditional model as $p(y\mid \x, \theta) = p(y, \x \mid \theta)/p(\x\mid \theta)$ (as in e.g. \cite{Muller1996,Park2010}, among others). There is some cost to the joint modeling approach -- which itself has commanded a large literature (see \cite{West1993, Muller1996} and variations in \cite{Shahbaba2009a, Taddy2010, Molitor2010, Wade2011,Dunson2010, Hannah2011,Wade2014}) -- in terms of computation and accuracy as the dimension of the covariate vector grows (see \cite{Hannah2011,Wade2014}). The posterior can also depend on the distribution of the covariates, even when care is taken to separate the parameter spaces in the prior, as the auxiliary joint model assumes a common clustering for the response and the covariates (see Griffin's discussion of \cite{Dunson2010}; also, \cite{Walker2013} and \cite{Wade2014}). Thus, other nonparametric Bayesian models focus explicitly on the conditional distributions of interest; these date to (at least) MacEachern's seminal work on dependent Dirichlet processes (DDPs) \citep{Maceachern1999,MacEachern2000}. A DDP is a prior for a collection of distributions such that at each covariate value the process is marginally a DP. Models in this class include \cite{DeIorio2004,Griffin2006,dunson2008bayesian,DeIorio2009,Wang2011bayesian}, and numerous other specializations to spatiotemporal or hierarchical models. \cite{Barrientos2012} characterize the DDP in terms of copulas and provide results about its support and about kernel mixtures using the DDP.

\section{Heteroscedastic Regression with BART priors}\label{sec:hetbart}

\subsection{Bayesian Additive Regression Trees (BART)}\label{sec:bart}

Bayesian Additive Regression Trees (BART) were introduced by \cite{Chipman2010} (henceforth CGM) as a nonparametric prior over a regression function $f$ designed to capture complex, nonlinear relationships and interactions. Specifically, for observed data pairs $\mathcal{D} = \{(y_i, \x_i); 1\leq i\leq n\}$, CGM propose the regression model:
\begin{equation}
y_i = f(\x_i) + \epsilon_i,\quad \epsilon_i\iid N(0,\sigma^2)\label{eq:meanbart}.
\end{equation}
The BART prior represents $f$ as the sum of many piecewise constant regression trees.
%Each of the trees in BART can pick up `weak' patterns in the data, such that the sum of all trees provides a combined representation of many such signals. The particular appeal in our context is its ability to detect inohmogeneous  response surfaces including discontinuities \cite{hill2011bayesian} and to pick up interactions in the predictors, as well as its computational efficiency. 
%$\mathcal{T}=\left\{T_1,\ldots,T_L\right\}$. 
%, building off earlier tree-based regression models \citep{breiman2001random,evers2009locally}. 
%
Each tree $T_h,\;1\leq h\leq m$ consists of a set of interior decision nodes (where decisions are generally of the form $x_j<c$ for some value $c$) and a set of $b_h$ terminal nodes. The terminal nodes have associated parameters $\M_h = (\mu_{h1},\mu_{h2},\dots \mu_{hb_h})'$.  For each tree there is a partition of the covariate space $\{\mathcal{A}_{h1},\dots,\mathcal{A}_{hb_h}\}$ with each element of the partition corresponding to a terminal node. A tree and its associated parameters define step functions:
\beq
\g(\x, T_h, \M_h) = \mu_{hb}\text{ if }\x\in\mathcal{A}_{hb} \text{ (for $1\leq b\leq b_h$)}.
\eeq
%Each tree and its associated parameter vector define a piecewise constant function over covariate space: $g_h(x) = \mu_{hb}\ \text{if}\ x\in \mathcal{A}_b$. 
These functions are additively combined to obtain $f$:
\begin{equation}
f(\x) = \sum_{h=1}^m \g(\x, T_h, \M_h).
\end{equation}
This model has been shown empirically and theoretically (see, e.g. \cite{Chipman2010,artofbart}) to be accurate, highly flexible, and robust to the presence of irrelevant covariates. The default prior parameters work very well in practice and the model admits an efficient sampler. More details on BART can be found in the Supplement.

%, so as to describe a function $g_l(x)$ which is piecewise constant.  In other words, each tree defines a decision rule of the form $\left\{x\in A\right\}$ vs  $\left\{x\not\in A\right\}$, partitioning the covariate space into disjoint regions. Conditional on this partition, a constant mean function is fit to the data in each disjoint region.

\subsection{A BART prior for variance functions}
As mentioned in the introduction, to allow for flexible density regression in finite samples, it may be useful to allow the variance of the process to depend on $\x$ as well. To this end, we adopt the log-linear BART prior of \citet{murray2021log} for the log-variance. Specifically, consider the heteroskedastic regression problem $y_i = f(\x) + \sigma^2(\x)\epsilon_i$ where $\epsilon_i\overset{iid}{\sim} N(0,\sigma_0^2)$. \citeauthor{murray2021log} places a log-linear BART prior on $\sigma^2(\cdot)$:
$$\log[\sigma^2(\x)]= \sum_{h=1}^{m_v} g(\x, \vT_h, \vM_h),$$
where $\{(\vT_h$, $\vM_h)\}$ are trees and parameters for the variance function. The prior for the exponentiated leaf parameters $\exp(\mu_{hb}^{(v)})$ is conjugate, symmetric on the log scale, and can be calibrated to match the expected prior range of the log-variance process. More details on BART's extension to log-linear models can be found in the Supplement.

\section{Density Regression with BART priors}\label{sec:dr}

Even with a flexible variance function, the normality assumption of heteroscedastic BART may be too restrictive; for example, at any covariate value it yields symmetric predictive distributions. In this section, we extend the heteroscedastic BART model to general density regression problems by introducing a continuous latent variable $U$, which is treated as an omitted variable independent of $\x$. 
%Since the mean and variance functions have BART priors they are step functions, so this continuous latent variable will in fact induce a discrete mixture model for the predictive distribution. B
Before introducing the model in full generality we will motivate  the use of continuous latent variables in the density estimation setting.

\subsection{Continuous latent variable priors for a single density}\label{sec:single}

Consider the following generalized location model for estimating a single density:
\beq
p(y) = \int_0^1 \phi_\sigma(y - f(u))
\,du.\label{eq:con0}
\eeq
An equivalent representation in terms of a latent variable is
\beq
Y = f(U) + \epsilon\label{eq:transmodel},\quad U\sim U(0,1),\quad \epsilon\sim N(0, \sigma^2)
\eeq
where \eqref{eq:con0} is obtained on marginalizing over $U$.
%
%
%To motivate the full density regression model we begin with the class of transformation models for a single density. 
%\beq
%y_i = g(u_i) + \epsilon_i\label{eq:transmodel}
%\eeq
%where $u_i \sim U(0,1)$, $\epsilon_i\sim N(0, \sigma^2)$ are independent. The $u_i$ are latent variables, inducing a model for the density of $y_i$ through marginalization:
%\beq
%f(y\mid \sigma) = \int_0^1 \frac{1}{\sqrt{2\pi}\sigma}e^{-\frac{1}{2\sigma^2}(y-g(u))^2}d\,u.
%\eeq
In the limit as $\sigma\rightarrow 0$, $Y \,{\buildrel d \over =}\, f(U)$ where $U\sim U(0,1)$. This class of models can be quite broad, depending on the prior for $f$; if $f$ is the quantile function of a distribution $P$, then $Y\sim P$. We do not restrict $f$ to be monotone; while it would be possible to do so, this substantially increases the computational burden and since subsequent inference is on the induced density for $y$ (or $y\mid \x$ below) or its functionals rather than $f$ itself, the monotonicity constraint is not necessary.

Discrete location mixtures arise as a special case of this model when $f$ is a step function. Suppose $f(u) = \mu_h$ for $u\in [\nu_h, \nu_{h+1})$, where $\nu$ is an increasing sequence on $[0,1)$ such that $\nu_1=0$ and $\sum_{h=0}^\infty (\nu_{h+1}-\nu_h) = 1$. Then we have:
\[
p(y) = 
\int_0^1 \phi_\sigma(y - \mu_h)
%e^{-\frac{1}{2\sigma^2}(y-\mu_h)^2}
\ind{u\in [\nu_h, \nu_{h+1})}d\,u
= \sum_{h=1}^\infty (\nu_{h+1}-\nu_h)\phi_\sigma(y - \mu_h).
\]
This representation is intimately related to the augmented model used for slice sampling infinite mixture models \citep{Walker2007, Kalli2011}, where a prior on $f$ is induced via the prior on mixture component weights $\pi_h\equiv \nu_{h+1}-\nu_h$.

Priors on mixture weights are only one way to induce the prior on $f$. \cite{Kundu2014} proposed placing a Gaussian process prior directly on $f$, suggesting models centered on a prior guess of the quantile function and using a squared exponential covariance function. Theoretically, this is a flexible choice \citep[see][]{dunson_density, Kundu2014}; however, it introduces computational difficulties, requiring a discretization of the space that may reduce the quality of the subsequent inference. Furthermore, it is not immediately clear how to introduce multiple or categorical covariates into this framework.

The continuous latent variable model is appealing, however. In many contexts it is more intuitive to think of distributional features as arising from some omitted or latent continuous variables, and not from heterogeneity due to multiple independent subpopulations. Adapting BART to this setting yields priors for $f$ which incorporate covariates flexibly, are approximately smooth, and do not require discretization of the latent variable \emph{a priori}. 

%Using BART as a prior for $f$ removes the problem of selecting a grid; there is no computational benefit to discretizing $U$ at all. The data directly inform on how much the function varies in $U$, which is reflected in the tree structure.
%\cite{Pati2011} study an alternative model for density regression that replaces $g(u)$ with a more general function $g(u, x)$, but provide only theoretical analysis. This model avoids unnecesarily fitting a joint model, but the other difficulties remain. 
%Further, the prior is not invariant to transformations of the covariates and constructing covariance functions for mixed discrete and continuous covariates is not straightforward  (see \cite{Gramacy2009} and \cite{Broderick2011} for detailed discussion on this point). Finally, the use of a single scale parameter can be problematic in many realistic settings. 

\subsection{Density Regression with BART (DR-BART)}

\begin{figure}
\begin{center}
{\centering \includegraphics[width=.9\linewidth]{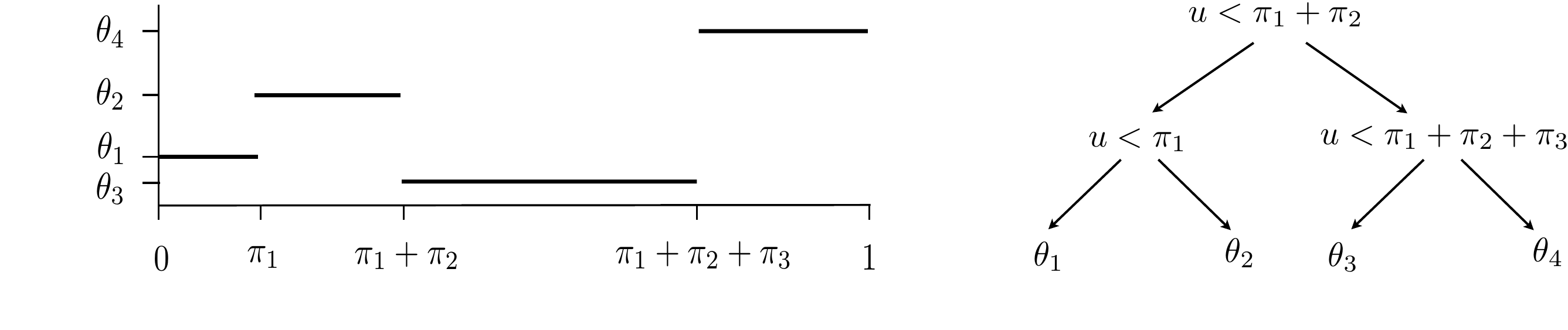} 
}
\end{center}
\caption{Representing a step function (left) as a binary tree (right).}
\label{fig:stepg}
\end{figure} 

\begin{figure}
\begin{center}
{\centering \includegraphics[width=.5\linewidth]{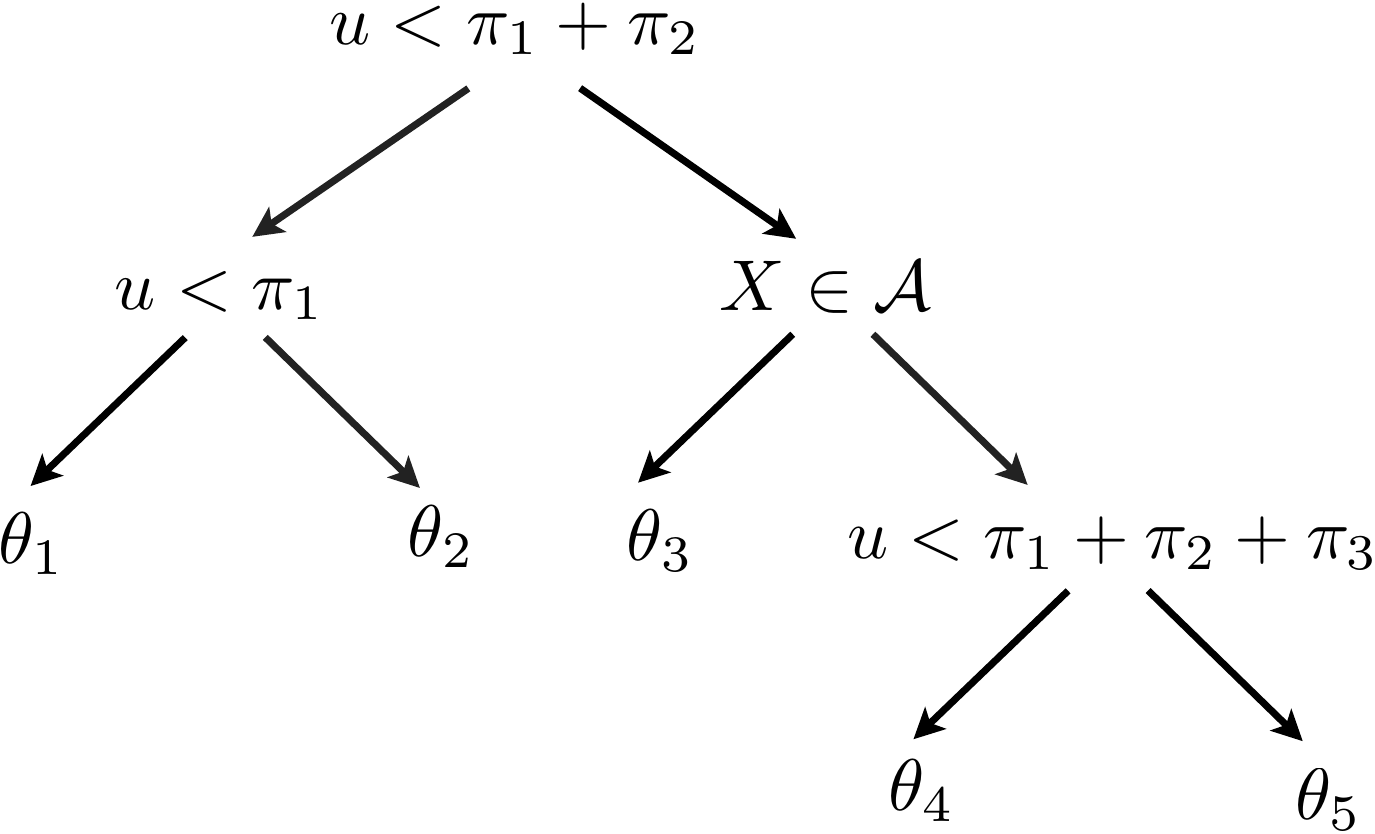} 
}
\end{center}
\caption{Modifying the tree in Fig. \ref{fig:stepg} to incorporate covariates}
\label{fig:stepfunsx}
\end{figure} 

To motivate the use of tree-based priors, recall that the normal location mixture model corresponds to a step function for $f$, which can also be represented as a binary tree (Fig. \ref{fig:stepg}). In the tree-based representation, incorporating covariates is simply a matter of introducing additional covariate-based splitting rules; Figure \ref{fig:stepfunsx} modifies the tree to split on $\x$. Marginalizing over the latent variable in this example gives the two conditional densities:
\begin{align}
p(y\mid \x\in A) =& \pi_1
\phi_\sigma(y - \theta_1)
+\pi_2\phi_\sigma(y - \theta_2)
+(\pi_3+\pi_4)\phi_\sigma(y - \theta_3)\\
p(y\mid \x\not\in A) =& \pi_1\phi_\sigma(y - \theta_1)
+\pi_2\phi_\sigma(y - \theta_2) + \pi_3\phi_\sigma(y - \theta_4) + \pi_4\phi_\sigma(y - \theta_5)
\end{align}

The resulting model has some interesting properties. There are shared components and shared weights (e.g. $\theta_1,\theta_2$ and $\pi_1, \pi_2$) but also components and weights that are unique to each conditional distribution (e.g. $\pi_3+\pi_4$ versus $\pi_3, \pi_4$). This allows borrowing of information across the covariate space, with the degree of borrowing controlled by the tree structure. Trees with multiple interior splits on $\x$ allow the model to capture multiscale structure, as the degree of borrowing varies across the covariate space. In the extreme case, a tree that splits on $\x$ before splitting on $U$ will yield two independent density estimates, while a tree that doesn't split on $U$ yields a standard Bayesian CART model.

In practice, a single tree model will probably be inadequate for many applications. With many covariates, nearly smooth mean functions, or highly skewed/multimodal distributions the tree will have to grow quite large. Additionally, the presence of the latent variable $U$ will tend to yield poor mixing and further complicate the design of MCMC algorithms for the single tree model (see \cite{Chipman1998,Denison1998} for discussion of the complications arising in MCMC for single tree models \emph{without} latent variables).

A BART prior for $\f$ is a natural alternative to a single tree model and additional flexibility can be obtained by modeling the variance function as well. For density regression we modify both BART priors from Section \ref{sec:hetbart} to include a latent variable $U$. %Before describinThe latent variable formulation of most general DR-BART model is as follows:
The most general model for density regression with BART (DR-BART) is:
\begin{align}
U\sim U[0,1],&\quad \epsilon\sim N(0, \sigma_0^2)\\
Y = \f(\x, U) &+ \exp[v(\x, U)/2]\epsilon.
\end{align}
Letting $\sigma(\x, u) = \exp[v(\x, u)/2]\sigma_0$, the density function at $\x$ is
\beq
p(y\mid \x) = \int_0^1 \phi_{\sigma(\x, u)}(y - f(\bx, u))\,du,\label{eq:con}
\eeq
a location-scale mixture of normals. In addition to the full DR-BART model, two reduced models are potentially of interest:
\begin{enumerate}
\item  Location-only mixture (DR-BART-L): Constant bandwidth, $\sigma(\x, u)\equiv \sigma_0$
\item  Location mixture with heteroscedasticity (DR-BART-LH): Covariate-dependent \\bandwidth parameter, $\sigma(\x, u)\equiv \sigma(\x)=\exp[v(\x)/2]\sigma_0$
\end{enumerate}
Priors for the DR-BART parameters are specified as follows:
\begin{itemize}
\item $f$, the location function: The CGM BART prior with $k=2$ and $m = 250$ trees. We also require that the leaves of each tree contain at least 5 observations. In the context of DR-BART this condition is related to priors on mixture models for a single density that require all the components to be occupied (see e.g. \cite{Diebolt1994}).
\item $v$, the bandwidth function: The BART variance prior from Section \ref{sec:hetbart}, with 100 trees and $a_0$ (a hyperparameter for $\sigma(\x,u)$ defined in the Supplement) calibrated to a reasonable range, as described in the Supplement.

\item Depending on the model, $\sigma_0 ^ 2$ is either fixed (DR-BART, DR-BART-LH) or given a further inverse gamma prior (DR-BART-L): $\sigma_0 ^ 2 \sim IG(\nu_0/2, \nu_0\xi_0/2)$, with $\xi_0$ a prior guess at an appropriate bandwidth. In DR-BART-L, posterior inferences are insensitive to reasonable specifications of $\xi_0$, but allowing $\nu_0 \rightarrow 0$ puts too much mass near 0. Since $\sigma_0 ^ 2$ is not identified separately from the flexible variance function $\nu$ in the full DR-BART/DR-BART-LH models, the prior should be very informative and in practice fixing it at a sensible guess (some fraction of the sample variance or of the variance of the OLS residuals) seems to work well. It could also be elicited more formally. 

\end{itemize}
Results appear to be more or less insensitive to the numbers of trees in $f$ and $v$ provided they are large enough, and the values chosen here reflect experience with mean regression BART and the belief that variance functions are less complex than location functions. 

Before describing posterior sampling we will describe the properties of $f$ and $v$ and provide some intuition for their roles in the model.

\subsubsection{The location function $f$}

With a BART prior for $\f$, $\f(\x, u)$ decomposes as
\begin{align}
\f(\x, u) &= \sum_{h=1}^m g((\x, u), T_h, \M_h) \equiv \f_x(\x) + \f_{xu}(\x, u) + \f_u(u).\label{eq:decomp}
\end{align}
The three functions in \eqref{eq:decomp} are defined (from left to right) as the sum of the trees splitting only on $\x$, on both $\x$ and $u$, and only on $u$. These terms can capture covariate effects that are pure location shifts, covariate-dependent distributional features, and distributional features common across covariate space (respectively). Using posterior samples of the trees to try to infer which variables influence the responses and in what manner is somewhat fraught, however. For example, it is possible that a tree in $\f_{xu}(\x, u)$ split trivially on $U$ in that the leaf parameters on either side of the split are nearly identical or cancelled by the contribution of other trees. This is essentially the same difficulty reported by CGM in doing variables selection in BART mean regression by counting the trees splitting on a particular variable.  But \eqref{eq:decomp} does give some insight into how the model \emph{can} capture the complex effects that $\x$ may have on the distribution of $Y$, even with a constant bandwidth.

Since $\f$ and $v$ are step functions, this model is equivalent to a discrete mixture of normal distributions. But the prior is much different than the usual priors in covariate-dependent mixture models. First, the number of distinct mixture components with positive probability varies across covariate space like in the single tree model. Second, unlike in the single tree model, the components are correlated \emph{a priori}; given a fixed set of $m$ trees $T$ we have:
\begin{align}
\Cov(\f(\x, u),\f(\x', u')) &= \sigma^2_\mu \mathcal{N}[(\x, u), (\x', u')]\label{eq:fcov}\\
\Cor(\f(\x, u),\f(\x', u')) &= \mathcal{N}[(\x, u), (\x', u')]/m\label{eq:fcor}
\end{align}
where $\mathcal{N}[(\x, u), (\x', u')]$ is the number of trees where $(\x, u)$ and $(\x', u')$ are in the same leaf. It follows from  \eqref{eq:fcor} that if $|u-u'|>|u-u''|$, then
\begin{align}
\Cov(\f(\x, u),\f(\x, u'))&\geq\Cov(\f(\x, u),\f(\x, u''))\\
\Cor(\f(\x, u),\f(\x, u'))&\geq\Cor(\f(\x, u),\f(\x, u'')). 
%\Cov(\f(\x, u),\f(\x', u))\geq\Cov(\f(\x, u),\f(\x', u'))
\end{align}
%and identical expressions for the correlation. Further, if $\x'$ and $\x''$
Given the potentially strong correlation, it is misleading to think of the steps in $f$ as ``mixture components'' in the usual sense.

\subsubsection{The bandwidth function $v$}

A covariate-adaptive bandwidth parameter will often be important in these models. Since $U$ and $\epsilon$ are independent, in the DR-BART-L model $\Var(Y\mid \x) \geq \Var(\sigma_0\epsilon) = \sigma_0^2$. On the other hand, in the DR-BART-LH model $\Var(Y\mid \x) \geq \Var(\exp[v(\x)/2]\epsilon) = \sigma(\x)^2$. In a model without a covariate adaptive bandwidth, $\sigma^2_0$ must be at least as small as the most concentrated predictive density to avoid oversmoothing, yielding much rougher density estimates elsewhere. The efficiency of the MCMC sampler suffers as well, as smaller bandwidths imply more concentrated distributions for $U_i$.  A covariate dependent bandwidth might be preferable for this reason, even if a single bandwidth seems like a reasonable simplification.

The full DR-BART model has an additional degree of flexibility due to scale mixing. It can allow $v(\x, u)$ to grow large in some areas of of $U$-space, effectively ``turning off'' portions of $\f(\x, u)$ or capturing relatively flat areas of the density. The behavior of the different models is easiest to understand with an example, presented in Section \ref{sec:sims}.

\subsubsection{Posterior Sampling}

Generating samples from the posterior with MCMC is straightforward. Conditional on values for the latent variables $u=(u_1, u_2, \dots, u_n)$, DR-BART reduces to the heteroscedastic BART model so that sampling for the other parameters proceeds as described in the Supplement. The latent $u_i$ have full conditionals
\beq
p(u_i\mid -) \propto \phi_{\sigma(\x, u)}(y - f(\x, u))\ind{u \in \mathcal{B}(T, \vT)},\label{eq:ufc1}
\eeq
where $\mathcal{B}(T, \vT)$ is the set of possible $u$-values; that is, those that do not yield trees with leaves having fewer than 5 observations.  Since $f$ and $\sigma$ are step functions, \eqref{eq:ufc1} is piecewise constant, so $u_i$ can be updated with a Gibbs step: If $u^*_1 < u^*_2 < \dots  < u^*_k$ are the points of discontinuity of \eqref{eq:ufc1} and $u^*_0=0,\ u^*_{k-1}=1$, the Gibbs step first samples an interval from
\beq
\Pr(u_i\in (u^*_h, u^*_{h+1})\mid -) \propto 
\phi_{\sigma(\x, \tilde u_h)}(y - f(\x, \tilde u_h))\ind{ u \in \mathcal{B}(T, \vT)},\label{eq:ufcint}
\eeq
where $\tilde u_h = (u^*_{h+1} - u^*_h)/2$ (or any other point in the interval) and then $u_i$ is sampled uniformly from the selected interval. Note that $\ind{u \in \mathcal{B}(T, \vT)}$ necessarily equals either 0 or 1 on the entire interval $(u^*_h, u^*_{h+1})$ for each $h$ by construction. 

While conceptually simple, the Gibbs sampling update can require a large number of likelihood evaluations ($k$ is often well into the hundreds). 
%It is impractical to compute the discontinuities in $f$ and $v$ for each $\x_i$ individually. Instead, one collects from all the trees all the splitting rules for $U$ to find $u^*_1,u^*_2,\dots u^*_k$ resulting in the finest possible partition, and the largest number of likelihood evaluations. 
On the other hand, a Metropolis step is difficult to tune because \eqref{eq:ufc1} is in general multimodal. An efficient alternative that doesn't require tuning is slice sampling which introduces latent variables $\omega_i$ so that
\beq
p(\omega_i,u_i\mid -) \propto \ind{
\omega_i < \phi_{\sigma(\x, u)}(y - \f(\x, u)}\ind{u \in \mathcal{B}(T, \vT)}.\label{eq:ufc2}
\eeq
Sampling proceeds using the techniques developed in \cite{Neal2003}. The slice sampler is much more efficient, which tends to make up for any loss in theoretical efficiency or mixing. 

\section{Theory}\label{sec:theory}
Here, we present some properties of the DR-BART model, showing that our proposed prior generates trees that are almost surely finite and upper bounding the rate at which DR-BART estimates the conditional density $p(y \mid \x)$. We focus on the special case where the predictors $\x_i\in [0,1]^p$ are continuous. This is a common assumption when studying theoretical properties of density regression models \citep[e.g.][]{pbss, Pati2013, modbart}. The assumption is even more innocuous here, as BART is invariant to monotone transformations of the covariates. Proofs can be found in the Supplement. 

\begin{thm}
\label{thm:finite}
Assume that the prior over a binary tree is as in CGM, but with a continuous uniform prior on splitting locations and no restriction to nonempty leaves. A tree sampled from this prior has finite depth with probability one.
\end{thm}
Thus, introducing a latent $u$ into the prior does not affect the finite depth of the trees. 
We now focus on providing upper bounds for the posterior concentration rate $\epsilon_n$ of the posterior. A rate $\epsilon_n \downarrow 0$ is said to be a \emph{rate of convergence} of the posterior with respect to a divergence measure $h$ if there exists a positive constant $M$ such that $\Pi\{h(p_0, p_{f,\sigma}) \geq M \epsilon_n \mid \Data_n\} \to 0$ in $F_0$-probability, where $\Data_n = \{(\bX_i, Y_i) : i = 1\ldots,n\}$ and $(\bX_i, Y_i) \iid F_0(d\bx, dy) = p_0(y \mid \bx) \ dy \ F_\bX(d\bx)$.

The conditional density of $y$ induced by our model is given by the convolution
\begin{align*}
  p_{f,\sigma}(y\mid\x) =
  \int
  \phi_{\sigma(\x,u)}(y - f(\x, u)) \ du,
\end{align*}
and the limit $\sigma(\x,u) \to 0$ is associated with a random variable with quantile function $f(\x,u)$ if $f(\x,u)$ is monotonically increasing in $u$. This suggests that a reasonable strategy for establishing that $p_{f,\sigma}$ is close to $p_0$ is to show that $f(\x,u)$ is close to the true conditional quantile function $f_0(\x,u)$.

We characterize the concentration of DR-BART with respect to the integrated Hellinger distance, defined: $h(p, q) = \big(\int (\sqrt{p(y|\bx)} - \sqrt{q(y|\bx)}) ^ 2\ dy \ F_\bX(d\bx)\big) ^ {1/2}$. Two other divergence measures which will be useful for us are the \emph{generalized Kullback-Leibler divergences} $\KL(p_1 \| p_2) = \int p_1 \log(p_1 / p_2) \ dy \ F_\bX(d\bx)$ and $V(p_1 \| p_2) = \int p_1 \log^2(p_1 / p_2) \ dy \ F_\bX(d\bx)$.

To study the posterior concentration of DR-BART, we make use of results from (i) \cite{artofbart} involving the concentration of BART in a regression setting and (ii) \cite{dunson_density, gp} who leverage similar results about Gaussian processes to show convergence in a latent variable model similar to the one considered here. Our proof extends both of these works in fundamental ways. We extend \cite{artofbart} by introducing latent variables into the tree structure. While having immediate implications for DR-BART this also lays the groundwork for concentration results for latent variable BART models that we will consider in the future. Compared to \cite{dunson_density}, we not only introduce covariates into the setup, but also place a non-trivial DR-BART prior over the bandwidth, compared to their choice of a parametric, covariate-independent prior. 

Next, we outline conditions required for our proof. Throughout, we will write $a \lesssim b$ to mean that there exists a positive constant $C$, possibly depending on $p_0$ and on hyperparameters but otherwise independent of $n$, $p$, or any other variables, such that $a \le Cb$.

\begin{paragraph}{Condition F (on $p_0$)}
  We assume that $\log p_0(y \mid \x)$ is $\alpha$-H\"older as a map from $[0,1]^{p+1}$ to $\mathbb R$ for some $0 < \alpha \le 2$. Additionally, we assume $\log p_0$ is $d_0$-sparse in the sense that it depends on $(y,\x)$ only through the coordinates in $S_0 \subseteq \{1,\ldots,p+1\}$ where $|S_0| = d_0$,  $d_0 = o(\log(n))$, and $d_0\log p = o(n)$. Lastly, we assume that $||f_0||_\infty \lesssim \sqrt{\log n}$.
\end{paragraph}

\begin{remark}
  The assumption that $\log p_0$ is $\alpha$-H\"older implies that $p_0$ is bounded and bounded away from $0$. Condition F is used both to ensure that $p_0$ can be well-approximated with convolutions and to ensure that $\|f - f_0\|_\infty$ is small with sufficiently large prior probability. 
\end{remark}

% Next we describe our assumptions on the prior $\Pi$ on $(f, \sigma)$.
\begin{paragraph}{Condition P (on $\Pi$)}
  Let $S \subseteq \{1,\ldots,p+1\}$ denote the coordinates of $\x$ which the trees $f(\x, u)$ and $v(\x, u)$ split on.
  \begin{itemize}
  \item[(P1)] The support set $S$ of $(f, \sigma)$ has prior $\pi(S) = \binom{p+1}{D}^{-1} \pi_D(D)$ where $D \equiv |S|$ and $\pi_D(d)$ is an \emph{exponentially decaying prior} satisfying $a_1 (p+1)^{-a_3} \pi_D(d - 1) \le \pi_D(d) \le a_2 (p+1)^{-a_4} \pi_D(d-1)$ for some positive constants $a_1,a_2,a_3,a_4$ and $d = 1,\ldots, p+1$.
  \item[(P2)] Given $S$, each tree $T_h, h = 1,\ldots,m$ and $T_h^{(v)}, h = 1, \ldots, m_v$ is assigned the branching process prior with splitting proportion $q(d) = \nu^d$ for some $\nu \in (0, 1/2)$.
  \item[(P3)] The leaf node parameters $\mu_{hl}$ of $T_h$ are assigned independent $N(0, \sigma^2_\mu)$ priors.
  \item[(P4)] The log-variance function is given by $v(\x, u) = \xi + \sum_{h=1}^{m_v} g(\x, u ; T_h^{(v)}, M_h^{(v)})$ where $e^{-\xi} \sim \Gam(a_\sigma, b_\sigma)$ and $\mu_{hl}^{(v)} \iid \pi_v$ where $\pi_v$ is a strictly positive density supported on an interval $[-V,V]$.
  \item[(P5)] Splits in the tree ensemble can occur only at a number $b_n$ of candidate split points $\mathcal Z_n \subseteq [0,1]^{p + 1}$, which are selected from uniformly. Additionally, $\log b_n \lesssim \log n$.
  \item[(P6)] For each $n$ there exists a decision tree $\widehat T$ and leaf node values $\widehat M$ built from the candidate split-points in $\mathcal Z_n$ such that the regression tree $f^\star(\x, u) = g(\x, u; \widehat T, \widehat M)$ satisfies $\|f_0 - f^\star\|_\infty \lesssim (\log n / n)^{\beta/(2\beta + d_0)}$ where $\beta = \min\{\alpha,1\}$.
  \end{itemize}
\end{paragraph}

\begin{remark}
  Though P4 constrains the support of $\pi_v$, $\xi$ is unbounded, allowing the variance function to have arbitrary scale even as the trees grant arbitrary flexibility in its shape.
\end{remark}

\begin{remark}
  The only assumption which is seemingly beyond our direct control is P6, which asserts that $f_0$ can be uniformly approximated with a single decision tree using the candidate split points $\mathcal Z_n$. \citet{artofbart} give several valid configurations of split points for which P6 would hold; for example, when $\mathcal{Z}_n$ is a regular grid of size $b_n \asymp n^{cp}$ for $c$ a sufficiently large constant if Condition F holds.
\end{remark}
Under these assumptions, we have the following theorem:

\begin{thm}
\label{thm:posterior_concentration}
Assume that Condition F and Condition P hold. Then, there exists a positive constant $M$ such that $\Pi\{h(p_0, p_{f,\sigma}) \ge M \epsilon_n \mid \Data_n\} \to 0$ in $F_0$-probability, where $\epsilon_n = (n / \log n)^{-\frac{\alpha}{\alpha+1}\times\frac{\beta}{2\beta+d_0}} + \sqrt{d_0\log(p + 1) / n}$, where $\beta = \min\{\alpha, 1\}$.
\end{thm}

\begin{remark}
  The gap between what is attainable by DR-BART-LH is larger than might be expected, as the rate is removed by a factor of $\alpha/(\alpha+1)$ from the minimax rate. The reason this occurs is that Condition F implies $f_0(\x, u)$ has H\"older smoothness of $\alpha+1 > 1$ in $u$, whereas BART is not known to be able to adapt to smoothness levels higher than 1. If we modify Condition F to state that $f_0(\x,u)$ is $\alpha+1$-smooth as a function of $(u,\x)$ (as opposed to just in $u$) then it is possible to show that replacing the BART model with the SBART model of \citet{sbart} gives the rate $n^{\alpha/(2\alpha + d_0)} \log(n)^{-\alpha(d_0+1)/(2\alpha+d_0)} + \sqrt{d_0 \log(p+1) / n}$ adaptively over $\alpha$ and $S_0$ for $\alpha \in (0, 2]$; extending these results to higher $\alpha$ using results of \citet{plummer2021statistical} is deferred to future work.
  % In Section~\ref{sec:sbart} we show that a variant of DR-BART based on the SBART model can be used to get a near-minimax-optimal convergence rate.
\end{remark}

As argued by \citet{modbart}, the posterior rate of convergence is (bounded by) a sequence $\epsilon_n$ with $n\epsilon_n^2 \to \infty$ if we can find positive constants $C_1, \ldots, C_4$ such that, for every sufficiently large $n$, there exists a set $\mathcal G_n$ of conditional densities satisfying the following:

\begin{itemize}
\item[(G1)] $\Pi\{p_{f, \sigma} \in \KL_{p_0}(C_1 \epsilon_n)\} \ge \exp\{-C_2 n \epsilon_n^2\}$, with $\KL_{p_0}(\epsilon) = \{p: \KL(p_0 \| p), V(p_0 \| p) \leq \eps ^ 2\}$
\item[(G2)]
  $\Pi(\mathcal G_n^c) \le C_3 \exp\{-(C_2 + 4) n \epsilon_n^2\}$.
\item[(G3)]
  $\log N(\mathcal G_n, \bar \epsilon_n, h) \le C_4  n \epsilon_n^2$ where $\bar \epsilon_n$ is a constant multiple of $\epsilon_n$ and $N(\mathcal G_n, \epsilon, d)$ denotes the $\epsilon$-covering number of $\mathcal G_n$ with respect to $d$ \citep[see, e.g.][]{ghosal_ghosh_vdv}.
\end{itemize}

The following two lemmas proved in the course of establishing Theorems 2 and 5 of \citet{artofbart} play a key role in establishing our results. The first ensures the prior on $f$ places sufficient mass around $f_0$ and is essential in establishing G1.
\begin{lemma}
  \label{lem:thick}
  Suppose that Condition F and Condition P hold and let $\delta_n = (\log n / n)^{\beta/(2\beta+d_0)}$. Then for sufficiently large $n$ we have
  \begin{align*}
    -\log \Pi(\|f - f_0\|_\infty \le \delta_n \mid S = S_0)
    \lesssim n \delta_n^2.
  \end{align*}
\end{lemma}
The second lemma ensures that the support of the BART prior is ``small'' in a suitable sense, allowing us to verify G2 and G3.

\begin{lemma}
  \label{lem:sieve}
  Let $\mathcal F$ denote the collection of decision tree ensembles with $m$ trees which (i) split on no more than $d$ variables, (ii) have at most $K$ leaf nodes per tree, (iii) have at most $b_n$ candidate split points, and (iv) satisfy $\sup_{hl} |\mu_{hl}| \le U$. Then
  \begin{align*}
    \log N(\mathcal F, \epsilon, \|\cdot\|_\infty)
    \lesssim
    d \log(p + 1) + K \log\left( \frac{d^m b_n^m K U}{\epsilon} \right).
  \end{align*}

\end{lemma}

\noindent Finally, Lemma \ref{lem:link} connects $h(p,q)$, $\KL(p\|q)$, and $V(p\| q)$ to the supremum norm.

\begin{lemma}
  \label{lem:link}
  Suppose that Condition F holds. Then there exists a constant $C_{\KL}$ independent of $(n,p)$ such that, for sufficiently small $\epsilon$, we have
  \begin{align*}
    \KL_{p_0}(C_{\KL} \epsilon)
    \supseteq
    \{\sigma(\x, u) \equiv \sigma \textnormal{ is constant},
    \sigma \in (\epsilon^{1/\alpha}, 2\epsilon^{1/\alpha}),
    \|f - f_0\|_\infty \le \epsilon^{1+1/\alpha})
    \}.
  \end{align*}
  Additionally, for any bounded measurable functions $f_1, f_2, \log \sigma_1, \log \sigma_2: [0,1]^{p+1} \to \mathbb R$:
  \begin{align*}
    h(p_{f_1,\sigma_1}, p_{f_2, \sigma_2})
    \lesssim
    \sqrt{\|\log \sigma_1 - \log \sigma_2\|_\infty} + \frac{\|f_1 - f_2\|_\infty}{\inf_{\x,u} \sigma_1(\x,u) \wedge \sigma_2(\x,u)}.
  \end{align*}
\end{lemma}

Straightforward application of Lemmas \ref{lem:thick} -- \ref{lem:link} suffices to verify G1 -- G3. The first part of Lemma \ref{lem:link} allows us to decompose the probability of the KL ball in G1 into pieces that can be bounded by Lemma \ref{lem:thick}. To verify G2, we consider a sieve defined by individual sieves for $f$, $v$, and $\xi$; the second part of Lemma \ref{lem:link} allows us to bound the Hellinger distance between a conditional density and an element in this sieve and Lemma \ref{lem:sieve} then ensures that the entropy is appropriately bounded. Given the sieve, G3 follows given Condition P. 

\section{Simulations}\label{sec:sims}
Here, we evaluate how well DR-BART and other methods estimate conditional densities and appropriately express uncertainty about these estimates. We do so via variants of a challenging univariate example. Additional simulation details can be found in the Supplement.

\subsection{Simulation 1: Contrasting DR-BART Models}\label{sec:sim_compare_drb}
Here, we introduce our basic simulation setup and gain insight into the differences between DR-BART-L, DR-BART-LH, and the full DR-BART model, before comparing to other methods. We consider a challenging example with a single regressor: 
% The true model is
%
\begin{gather}\label{eqn:DGP}
X_i\overset{iid}{\sim} U(0,1), \, 
Y_i = f_0(X_i) + \epsilon_i(X_i)
\end{gather}
where 
\[f_0(\mathbf{x}) = 5\exp[15(x-0.5)]/(1+\exp[15(x-0.5)]) - 4x\] and $\epsilon_i(X_i)$ is given by
\begin{gather*}
p(\epsilon \mid X=x) = \lambda(x)N(\epsilon; 2x-0.6, 0.3^2) + (1-\lambda(x))p_{G}(\exp(\epsilon), 0.5+x^2, 1.0)\exp(\epsilon), 
\end{gather*}
where $\lambda(x) = \exp[-10(x-0.8)^2]$ and the second component of $p(\eps \mid x)$ is a log-Gamma distribution with scale $1$ and shape $0.5+x^2$. Figure \ref{fig:drex-quant} shows selected quantile processes and conditional densities. The log-Gamma component is skewed and heteroscedastic, and the normal component is much more concentrated. The conditional distributions are nearly all unimodal, but for  $x$ values around 0.4 -- 0.6 the density is quite peaked around the mode with a heavy left tail. At $x=0.8$, $\epsilon$ has exactly a $N(1, 0.3 ^ 2)$ distribution.

\begin{figure}[h!]
\begin{center}
{\centering \includegraphics[height=.45\linewidth,width=.475\linewidth]{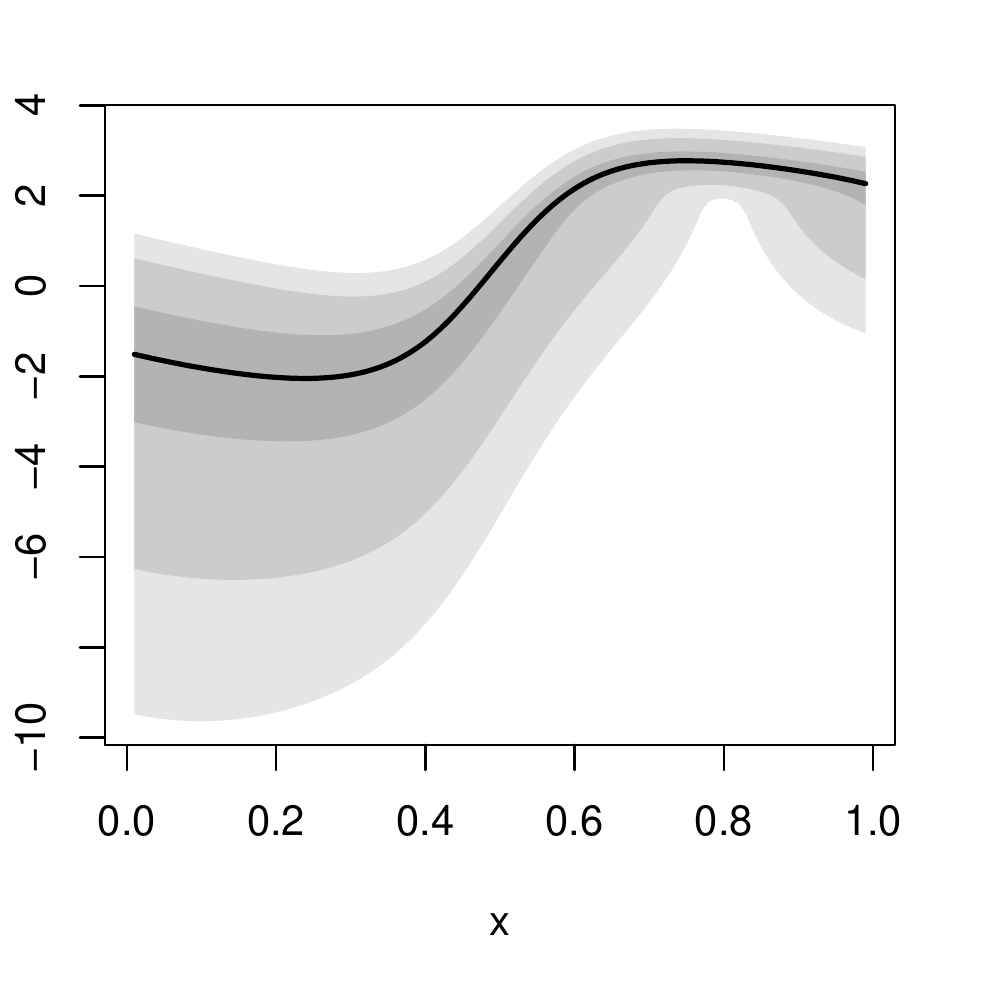} 
\includegraphics[height=.45\linewidth,width=.475\linewidth]{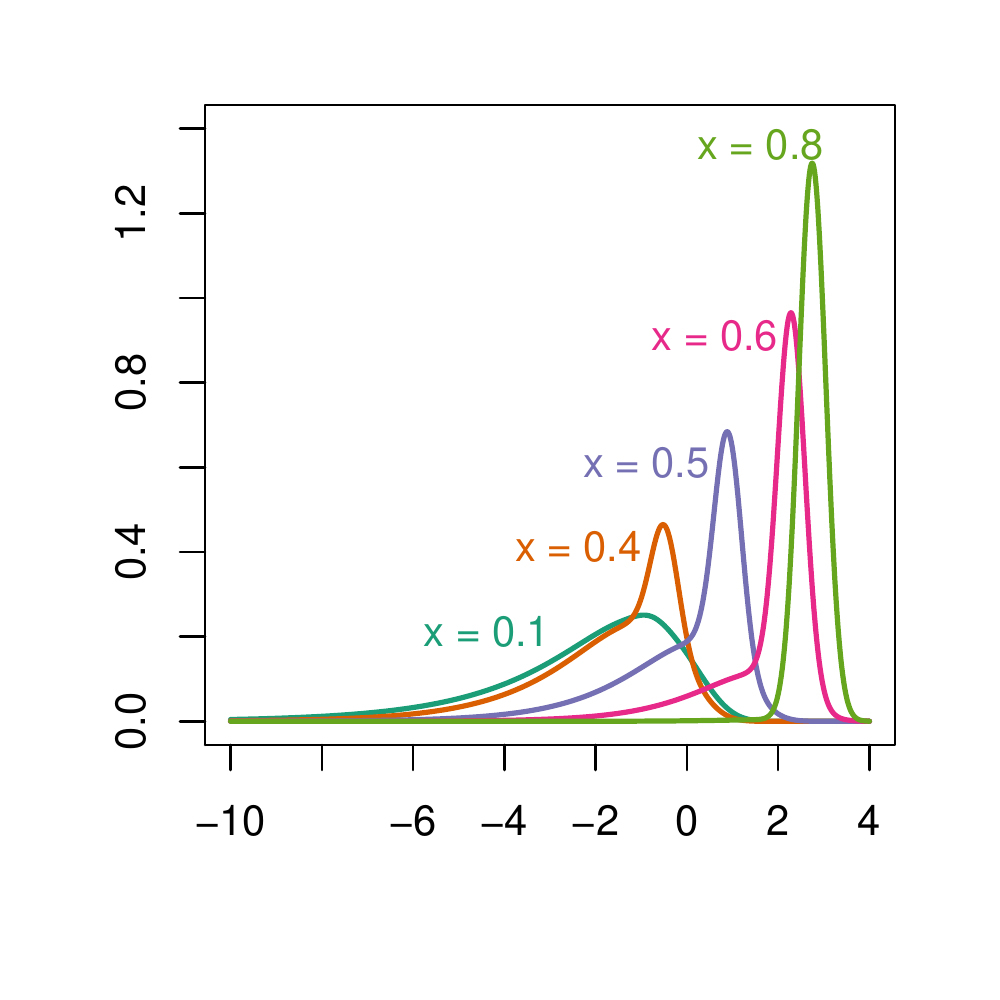} 
}
\end{center}
\caption{(Left) Conditional quantiles of the example in Section \ref{sec:sims}. Bands are equal-tailed 98\%, 90\%, and 50\% intervals as a function of $x$. The solid line is the median process. (Right) The conditional pdf for selected $\x$ values.}
\label{fig:drex-quant}
\end{figure} 

% We simulate $n = 800$ data points from the above model and fit all three DR-BART variants, collecting 5,000 samples after a burn in of 10,000.
Representative results for a single dataset are shown in Figure \ref{fig:drex-result}, which shows the estimated predictive density at $x = 0.1, 0.5, 0.8.$ All three models do well when $x = 0.8$; BART captures the nonlinear mean function particularly well since the noise is low and symmetric, and the variance is well-estimated since a peak of approximately this width is present in many of the conditional densities. However, DR- BART-L predictably struggles at $x = 0.1$ when the peak disappears and the density becomes much more diffuse and skewed. Some of this severe multi-modality can be mitigated by adjusting the prior to include more / larger trees. But the fundamental problem is with the single bandwidth parameter: it must be low to capture the peak, which means that the spread in $p(y | x = 0.1)$ has to be captured by $f(0.1,u)$ varying greatly in $u$, yielding rougher densities.

\begin{figure}[h!]
\centering 
\includegraphics[height=.9\linewidth,width=.95\linewidth]{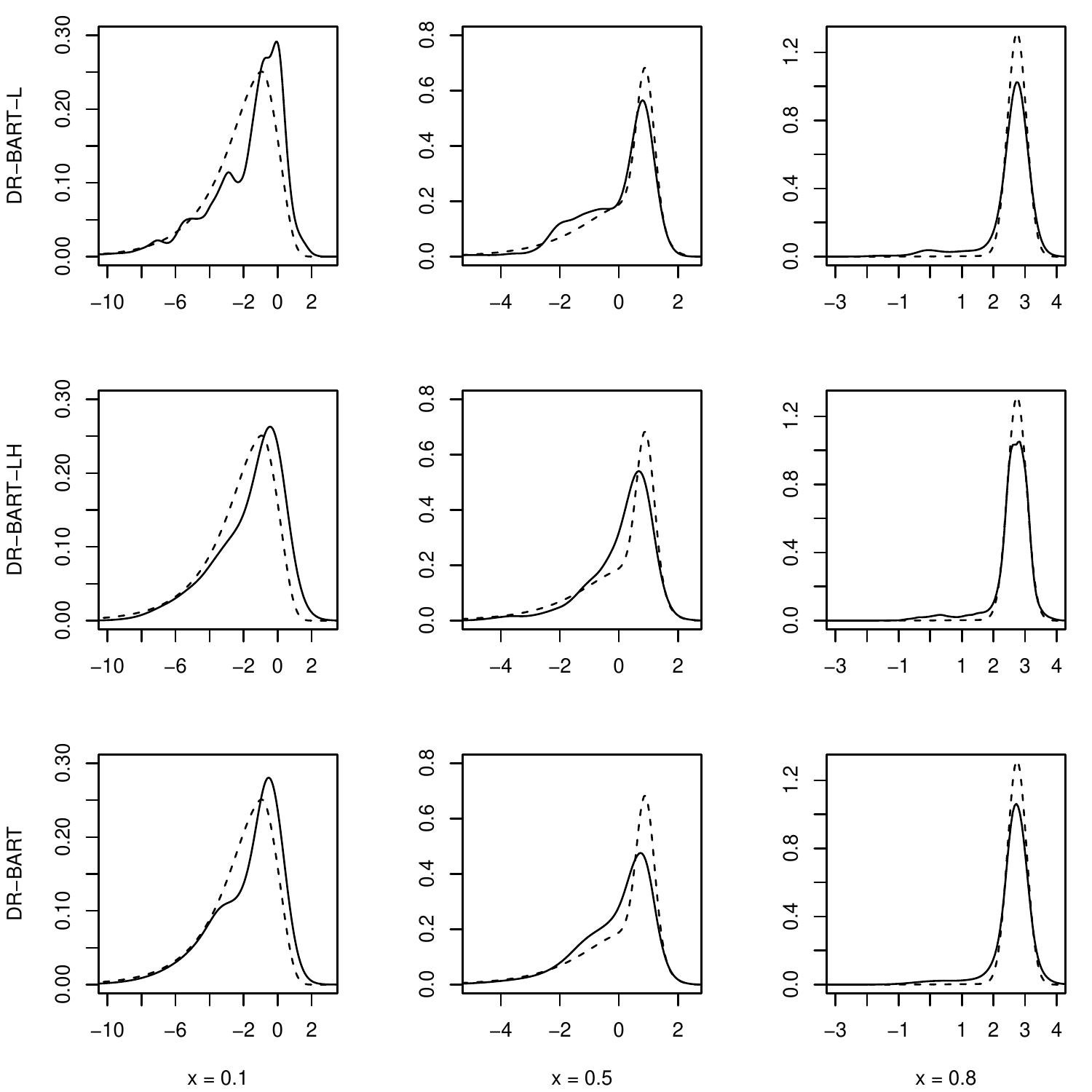} 
\caption{Representative true (dotted) and estimated (solid) densities from Simulation 1.}
\label{fig:drex-result}
\end{figure}

DR-BART-LH and DR-BART perform much better, with the full model doing slightly better, particularly when $x = 0.1$, since it is a location-scale mixture and can capture the long left tail by splitting on $u$ in the variance function to add high-variance “components”. The differences are fairly modest though, and the increase in computation time is {\raise.17ex\hbox{$\scriptstyle\sim$}}25\%. Given its comparable performance, DR-BART-LH is a viable alternative to the full model, especially when most of the conditional densities are not severely multimodal or skewed.

\subsection{Competing Methods}
We now compare the full \textbf{DR-BART} model to a variety of other methods for conditional density estimation. First, we compare to the Probit Stick-Breaking Process Mixture (\textbf{PSBPM}) in \cite{pbss}. Chung et al. model the conditional density $p(y|\mathbf{x})$ via an infinite mixture of normal linear regressions according to mixing distributions $P_x$:
\[p(y|\mathbf{x}) = \int \phi_{\tau}(y - \x'\beta)\,dP_{\mathbf{x}}\]
where the prior over the $P_{\mathbf{x}}$ is defined by a covariate-dependent stick-breaking prior.

Second, we compare to the Soft BART Density Sampler (\textbf{SBART-DS}) of \citet{modbart} which models conditional densities by modulating a base model $h(y|\mathbf{x}, \theta)$ via a link function $\Phi(\mu):$ $p(y|\mathbf{x}, \theta) \propto {h(y|\mathbf{x}, \theta)\Phi\{r(y, \mathbf{x})\}}$.
There are many possible choices for $h$, $r$, and $\Phi$; Li et al. choose to center their conditional densities on a normal linear regression: $h(y|\mathbf{x}, \theta) = \phi_{\sigma_\theta}(y - \x'\beta_\theta)$ and $\Phi$ to be a probit link. They take $r$ to be the weighted sum of Soft Additive Regression Trees \citep{sbart} with random Fourier expansions approximating a Gaussian process in the leaf nodes, as in \citet{tsbart}.

Third, we compare to a Dirichlet Process Mixture Model (\textbf{DPMM}) of \citet{dpmm} that models $(\mathbf{x}, y)$ as jointly normal and computes the implied conditional of $y$ on $\mathbf{x}$.

Lastly, we compare to a generalization of \citet{gp} that incorporates covariates into the transfer function $\mu$, which has a Gaussian process prior (\textbf{DR-GP}). This approach is similar to ours in its use of a latent $u$ to perform conditional density estimation. It differs in its use of 1. a Gaussian process to map $(u, \mathbf{x})$ to the observed response and 2. a homoscedastic, inverse Gamma prior on $\sigma$, instead of our BART priors on each.

\subsection{Simulation 2: Univariate Example}\label{sec:sim_uni}
\label{sec:sim1}
Here, we consider the same simulation as above. We simulate $n = 800$ data points from model \eqref{eqn:DGP} and fit all methods. Fits and credible intervals for a representative simulation are shown in Figure \ref{fig:sim1}. We see that DR-BART estimates the mean well with appropriate uncertainty quantification across all conditional densities, though it struggles to fully capture the peak of the normal density at $x = 0.8$ given how concentrated it is. Most the other methods are able to fit the densities at $x = 0.1$ and $x = 0.8$ reasonably well, but have difficulty with capturing both the peak and strong skew present in $p(y \mid x = 0.5)$. Results aggregated across 100 simulations are shown in Table \ref{tab:sim1}. As suggested by Figure \ref{fig:sim1}, DR-BART outperforms all other methods when $x \in \{0.1, 0.5\}$, but falls short of always capturing the peak at $x = 0.8$. DPMM estimates $p(y \mid x = 0.8)$ particularly well, which is to be expected given that it is built upon a normal specification.

We next assess how well the 95\% credible bands of each method cover the true $p(y \mid x)$. Table \ref{tab:sim1} also displays the proportion of simulations where the credible bands fully contain the true density within the true 95\% HDR interval. As Figure \ref{fig:sim1} suggests, DR-BART is able to do so consistently for $x = 0.1$ and $x = 0.5$, but not always for $x = 0.8$. For this and the next two simulations, the Supplement also contains information on credible band width and predictive coverage, which was consistently close to nominal for all methods but PSBPM.

\begin{figure}[h!]
    \centering
    \includegraphics[width=\textwidth]{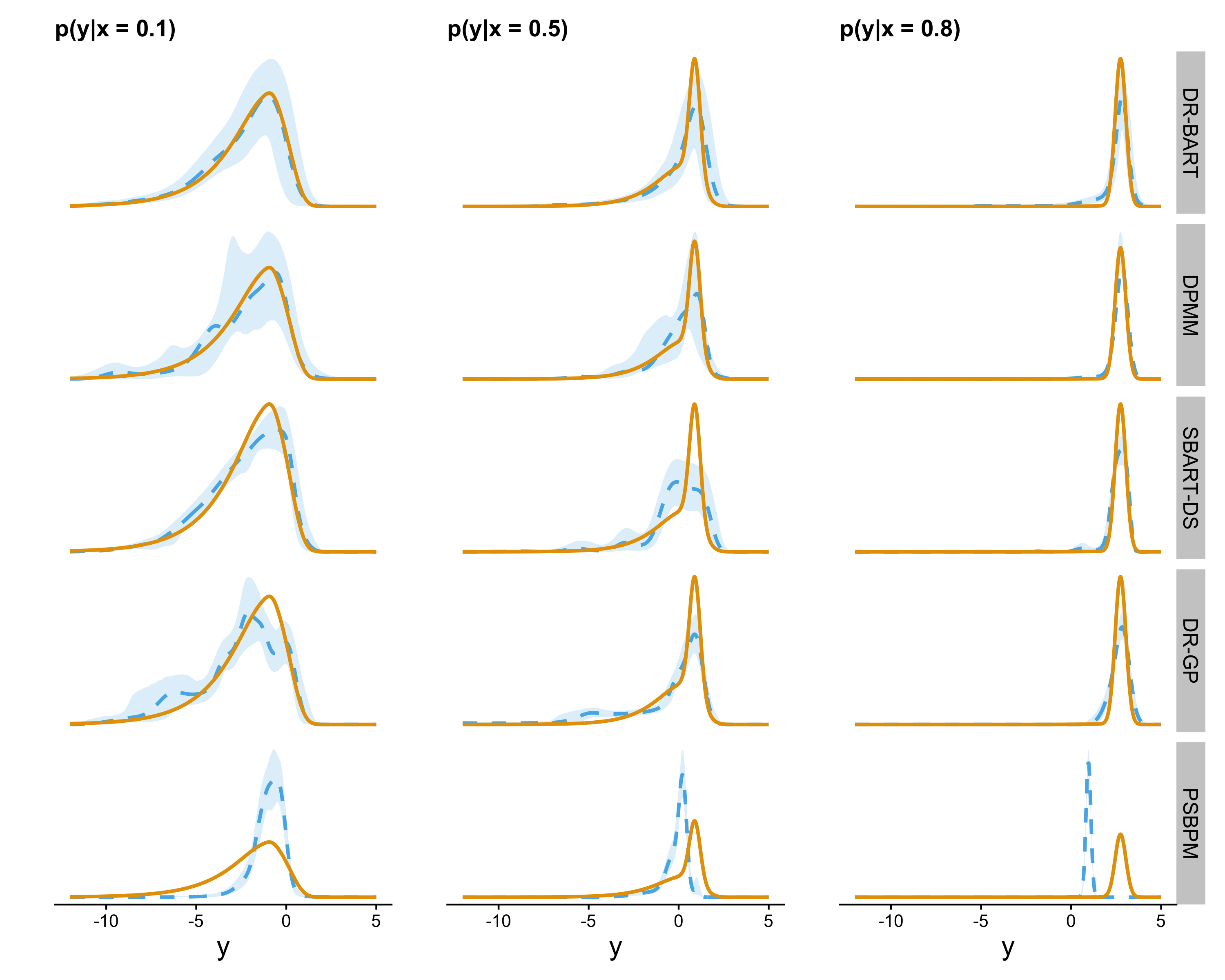}
    \caption{Representative true (solid) and estimated (dashed) densities from Simulation 2.}
    \label{fig:sim1}
\end{figure}

\begin{table}[hb]
\begin{adjustbox}{max width=1\textwidth,center}
\begin{tabular}{@{\extracolsep{1pt}} lcc|cc|cc|cc|cc} 
\toprule
& \multicolumn{2}{c}{\textsc{DR-BART}} & \multicolumn{2}{c}{\textsc{DPMM}} & \multicolumn{2}{c}{\textsc{SBART-DS}} &
\multicolumn{2}{c}{\textsc{DR-GP}} &
\multicolumn{2}{c}{\textsc{PSBPM}} \\
 [1ex] 
  & \textsc{Error} & \textsc{Coverage} & \textsc{Error} & \textsc{Coverage} & \textsc{Error} & \textsc{Coverage} &
  \textsc{Error} & \textsc{Coverage} & \textsc{Error} & \textsc{Coverage}  \\ 

\midrule
\large{$x = 0.1$ } & \large{1} & \large{0.93} & \large{1.67} & \large{0.75} & \large{1.65} & \large{0.48} & \large{2.74} & \large{0.02} & \large{9.80} & \large{0.00}\\
\large{$x = 0.5$} & \large{1} & \large{0.61} & \large{1.01} & \large{0.62} & \large{2.20} & \large{0.00} & \large{1.46} & \large{0.00} & \large{2.14} & \large{0.00}\\ 
\large{$x = 0.8$} & \large{1} & \large{0.19} & \large{0.28} & \large{0.85} & \large{0.79} & \large{0.00} & \large{1.03} & \large{0.00} & \large{2.27} & \large{0.00}\\ 
\bottomrule
\end{tabular} 
\end{adjustbox}
\caption{\footnotesize Error and Coverage for Simulation 2. Error is Wasserstein distance between true and estimated densities (normalized at each $x$ for interpretability), averaged over 100 simulations. Coverage is the proportion of simulations in which 95\% credible bands cover the truth within the true 95\% HDR region.}
\label{tab:sim1}
\end{table}

\subsection{Simulation 3: Irrelevant Covariates}\label{sec:sim_irrelevant}
Here, we generate data as above, but supply an additional 14 irrelevant, uniformly distributed variables, with pairwise correlations of 0.3, to each method. For each method, we evaluate the predictive density with the irrelevant variables fixed to $0.5$. Figure \ref{fig:sim2} and Table \ref{tab:sim2} illustrate the results. While the quality of DR-BART's fit decreases, it compensates by increasing the width of its credible bands to capture the truth. DPMM and DR-GP, which in the previous simulation fit better at $x = 0.1, 0.8$ do not, however, adjust their estimates of uncertainty to counter the substantial degradation in their fits. Perhaps most notable is that DR-BART still has good pointwise coverage, whereas the other methods are overconfident in their overfitting to the irrelevant covariates. We also note that the SBART-DS credible bands are not more erratic than in Simulation 1, reflecting BART's ability to perform variable selection. 

% As with the previous section, the appendix contains information on credible band width and predictive coverage; again, most methods had close to nominal coverage. 

\begin{figure}[h!]
    \centering
    \includegraphics[width=\textwidth]{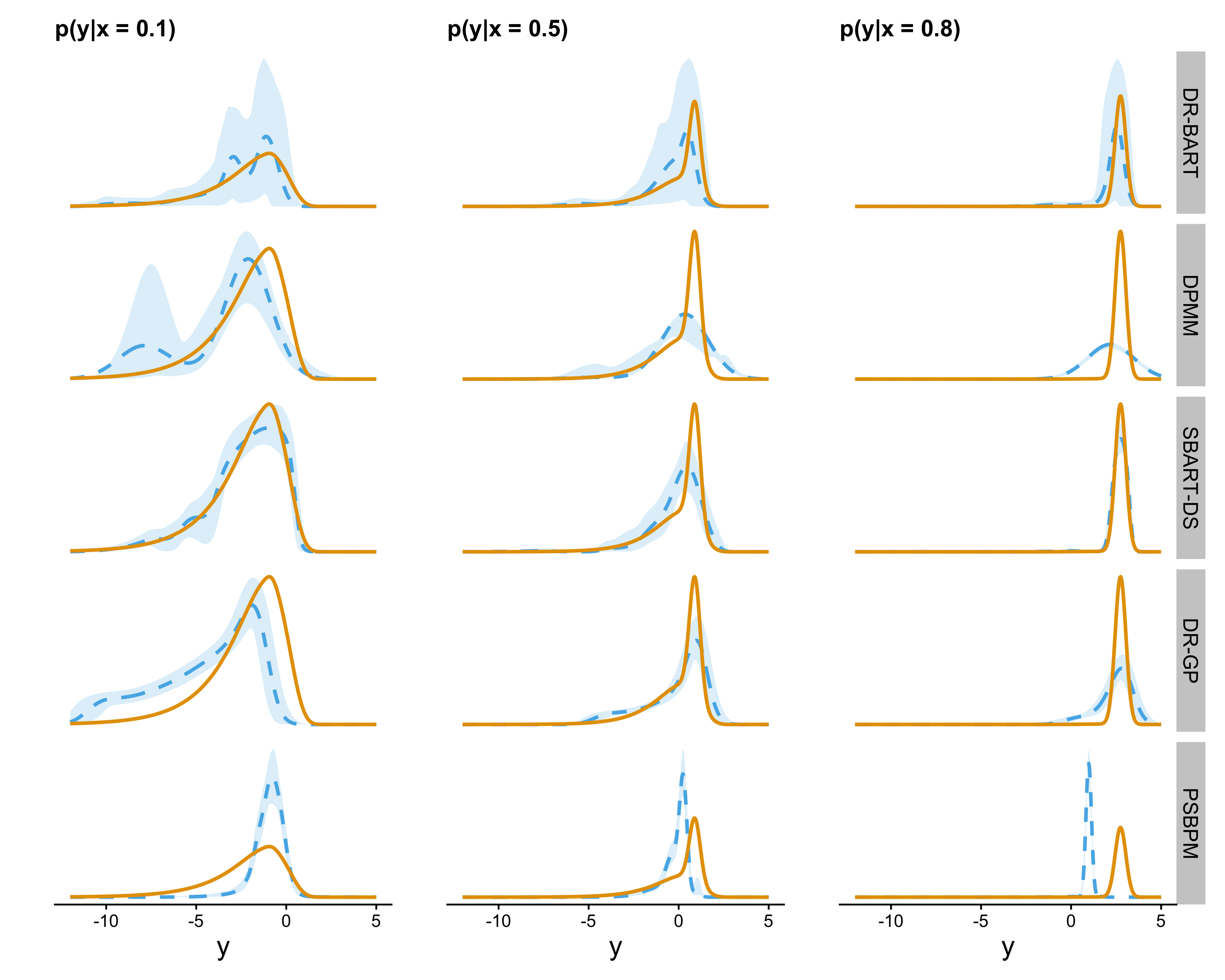}
    \caption{Representative true (solid) and estimated (dashed) densities from Simulation 3.}
    \label{fig:sim2}
\end{figure}

\begin{table}[ht]
\begin{adjustbox}{max width=1\textwidth,center}
\begin{tabular}{@{\extracolsep{1pt}} lcc|cc|cc|cc|cc} 
\toprule
& \multicolumn{2}{c}{\textsc{DR-BART}} & \multicolumn{2}{c}{\textsc{DPMM}} & \multicolumn{2}{c}{\textsc{SBART-DS}} &
\multicolumn{2}{c}{\textsc{DR-GP}} &
\multicolumn{2}{c}{\textsc{PSBPM}} \\
 [1ex] 
  & \textsc{Error} & \textsc{Coverage} & \textsc{Error} & \textsc{Coverage} & \textsc{Error} & \textsc{Coverage} &
  \textsc{Error} & \textsc{Coverage} & \textsc{Error} & \textsc{Coverage}  \\ 

\midrule
\large{$x = 0.1$ } & \large{1} & \large{0.78} & \large{1.10} & \large{ $0.00$ } & \large{0.78} & \large{0.46} & \large{2.63} & \large{0.00} & \large{5.58} & \large{0.00}\\
\large{$x = 0.5$}&\large{1} & \large{0.94} & \large{2.96} & \large{0.00} & \large{1.97} & \large{0.00} & \large{2.23} & \large{0.00} & \large{1.97} & \large{0.00}\\ 
\large{$x = 0.8$} & \large{1} & \large{1.00} & \large{1.60} & \large{0.00} & \large{0.56} & \large{0.02} & \large{2.52} & \large{0.00} & \large{1.70} & \large{0.00}\\ 
\bottomrule
\end{tabular} 
\end{adjustbox}
\caption{\footnotesize Error and Coverage for Simulation 3. Error is Wasserstein distance between true and estimated densities (normalized at each $x$ for interpretability), averaged over 50 simulations. Coverage is the proportion of simulations in which 95\% credible bands cover the truth within the true 95\% HDR region.}
\label{tab:sim2}
\end{table}

\subsection{Simulation 4: Insufficient Data}\label{sec:sim_insufficient}
Next, we explore the ability of DR-BART to appropriately express uncertainty about conditional distributions in regions of $x$-space where there is little data. To do so, we simulate $y$ given $x$ as in the previous experiments; however, instead of simulating $x$ uniformly, we simulate it from a mixture of uniforms: $0.475 \times U_{[0, 0.4]} + 0.05 \times U_{[0.4, 0.6]} + 0.475 \times U_{[0.6, 1]}$, where $U_{[a, b]}$ denotes the density of a $U(a, b)$ random variable. With only 5\% of the mass lying between $x = 0.4$ and $x = 0.6$, estimation of $p(y \mid x = 0.5)$ is much harder. With this experiment, we expect the performance of all methods to degrade; we are interested in whether their uncertainty grows appropriately to account for the lack of data. 

\begin{figure}[h!]
    \centering
    \includegraphics[width=\textwidth]{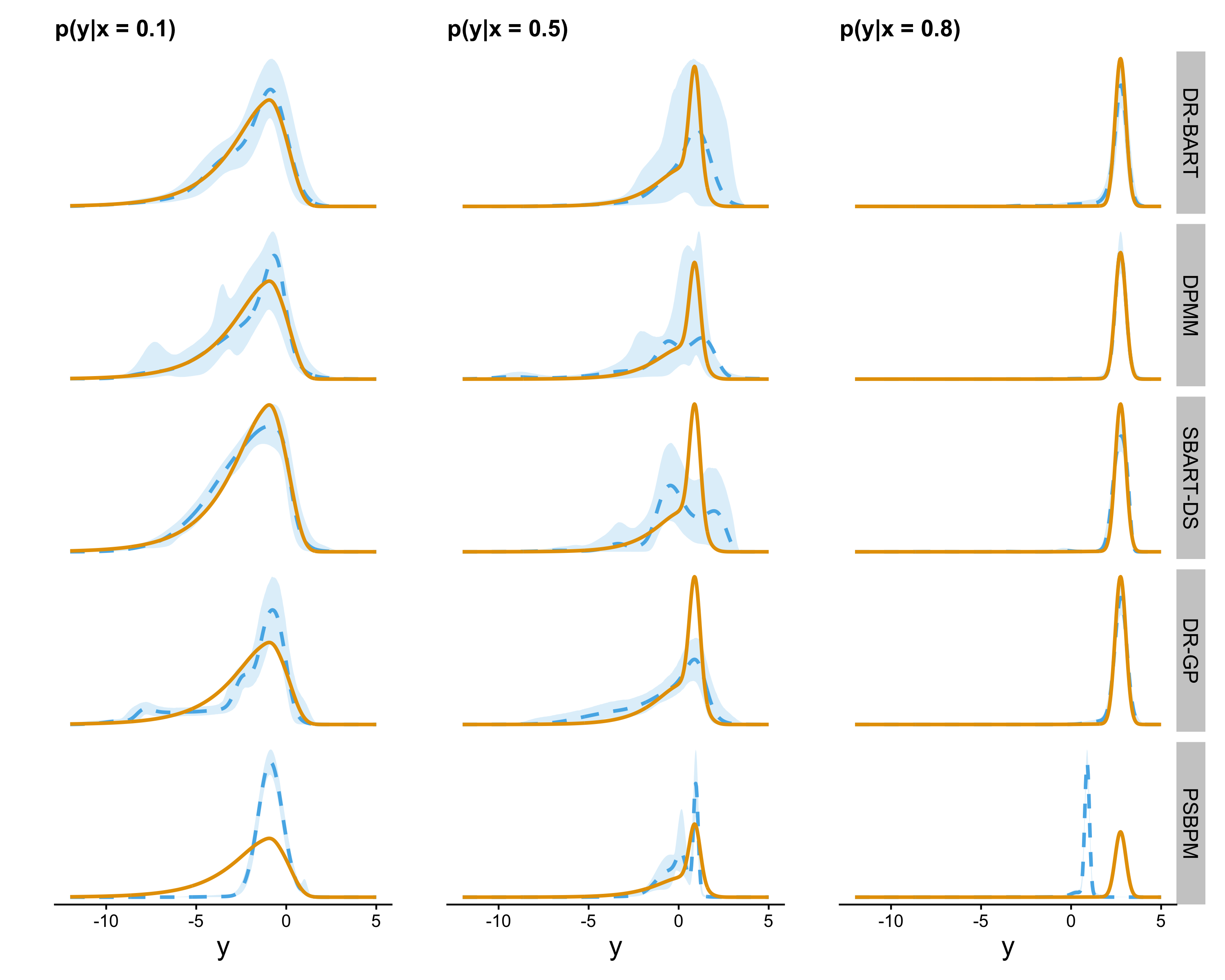}
    \caption{Representative true (solid) and estimated (dashed) densities from Simulation 4.}
    \label{fig:sim3}
\end{figure}

Figure \ref{fig:sim3} shows a representative simulation. Most of the methods perform better in estimating $p(y \mid x = 0.8)$, as there is now substantially more data in that region of the space. At $x = 0.5$, only DR-BART and DPMM increase their uncertainty in their estimation; the other methods have relatively poor fits with insufficient uncertainty about them. Table \ref{tab:sim3} summarizes this information across 50 simulations. DR-BART performs better in predicting $p(y \mid x = 0.5)$, but the relative performance for $x \in \{0.1, 0.8\}$ is comparable to that in Section \ref{sec:sim_uni}. Table \ref{tab:sim3} also supports the observation that DR-BART and DPMM are unique in increasing the uncertainty in their estimates in response to the lack of data. 

\begin{table}[t]
\begin{adjustbox}{max width=1\textwidth,center}
\begin{tabular}{@{\extracolsep{1pt}} lcc|cc|cc|cc|cc} 
\toprule
& \multicolumn{2}{c}{\textsc{DR-BART}} & \multicolumn{2}{c}{\textsc{DPMM}} & \multicolumn{2}{c}{\textsc{SBART-DS}} &
\multicolumn{2}{c}{\textsc{DR-GP}} &
\multicolumn{2}{c}{\textsc{PSBPM}} \\
 [1ex] 
  & \textsc{Error} & \textsc{Coverage} & \textsc{Error} & \textsc{Coverage} & \textsc{Error} & \textsc{Coverage} &
  \textsc{Error} & \textsc{Coverage} & \textsc{Error} & \textsc{Coverage}  \\ 

\midrule
\large{$x = 0.1$ }&\large{ $1$ }&\large{ $0.96$ }&\large{ $1.55$ }&\large{ $0.94$ }&\large{ $1.59$ }&\large{ $0.36$ }&\large{ $2.87$ }&\large{ $0.02$ }&\large{ 10.35 }&\large{ 0.00}\\
\large{$x = 0.5$ }&\large{ $1$ }&\large{ $0.54$ }&\large{ $1.34$ }&\large{ $0.56$ }&\large{ $1.91$ }&\large{ $0.02$ }&\large{ $1.34$ }&\large{ $0.00$ }&\large{ 1.26 }&\large{ 0.00}\\ 
\large{$x = 0.8$ }&\large{ $1$ }&\large{ $0.16$ }&\large{ $0.29$ }&\large{ $0.80$ }&\large{ $0.79$ }&\large{ $0.02$ }&\large{ $1.01$ }&\large{ $0.00$ }&\large{ 2.69 }&\large{ 0.00}\\ 
\bottomrule
\end{tabular} 
\end{adjustbox}
\caption{\footnotesize Error and Coverage for Simulation 4. Error is Wasserstein distance between true and estimated densities (normalized at each $x$ for interpretability), averaged over 50 simulations. Coverage is the proportion of simulations in which 95\% credible bands cover the truth within the true 95\% HDR region.}
\label{tab:sim3}
\end{table}

\subsection{Simulation 5: Comparison of BART Models}
Here, we further compare DR-BART to SBART-DS. Because both are built off of (S)BART models, we expect them to share many desirable properties, such as high flexibility and the ability to identify interaction effects and perform variable selection. SBART-DS, however, is centered on a base model, which may degrade the fit when the truth is far from the base model. Furthermore, there are cases in which prior information suggesting a reasonable base model may be lacking. To explore reliance on the base model, we run the original univariate simulation, but change the mean to be $f_0(x) = a(x - 0.5) ^ 2$. For small $a$, the mean is approximately linear for $x \in [0, 1]$, as specified by SBART-DS. But as $a$ increases, it becomes increasingly nonlinear and we expect performance to degrade. Using $n = 1500$, we run 50 simulations for each $a \in \{0, 1, 5, 15, 25, 35\}$, summarized in Table \ref{tab:bart_compare_wass}. For small $a$, both methods perform well. But the increasing nonlinearity of the mean as $a$ increases hinders SBART-DS' performance deteriorates, while DR-BART is unaffected.

\begin{table}[t]
\begin{adjustbox}{max width=1\textwidth,center}
\begin{tabular}{@{\extracolsep{1pt}} lcc|cc|cc|cc|cc|cc} 
& \multicolumn{2}{c}{$a = 0$} & \multicolumn{2}{c}{$a= 1$} & \multicolumn{2}{c}{$a= 5$} & \multicolumn{2}{c}{$a= 15$} & \multicolumn{2}{c}{$a= 25$} & \multicolumn{2}{c}{$a= 35$} \\
\toprule
& \textsc{DR-B} & \textsc{S-DS} & \textsc{DR-B} & \textsc{S-DS} & \textsc{DR-B} & \textsc{S-DS} & \textsc{DR-B} & \textsc{S-DS} & \textsc{DR-B} & \textsc{S-DS} & \textsc{DR-B} & \textsc{S-DS} \\
\midrule
  \multicolumn{1}{l|}{$x = 0.1$} &0.39 & 0.56 & 0.20 & 0.55 & 0.37 & 0.50 & 0.39 & 0.62 & 0.82 & 0.97 & 0.37 & 1.04\\ 
  \multicolumn{1}{l|}{$x = 0.5$} &0.76 & 1.09 & 0.74 & 0.85 & 0.89 & 0.91 & 0.39 & 1.50 & 0.71 & 1.57 & 0.80 & 2.20\\ 
  \multicolumn{1}{l|}{$x = 0.8$} &1.11 & 0.94 & 1.32 & 0.96 & 1.05 & 1.31 & 1.05 & 2.53 & 1.26 & 2.90 & 1.04 & 3.25\\ 
\bottomrule
\end{tabular} 
\end{adjustbox}
\caption{\footnotesize Wasserstein distance between true, estimated densities in Simulation 5, averaged over 50 simulations. As $a$ increases and the true mean becomes increasingly nonlinear, SBART-DS (S-DS) is hindered by its reliance on a linear base model, while DR-BART (DR-B) is unaffected .}
\label{tab:bart_compare_wass}
\end{table}

\section{Applications}\label{sec:ex}

We consider two applications of DR-BART. In both cases, we fit DR-BART-LH; the full DR-BART model gave similar results since both applications involve fairly well-behaved densities. For all models, $\sigma_0 ^ 2$ is set to be half the standard deviation of the OLS residuals and results seem to be insensitive to reasonable choices of this value. We also set $a_0 = \log(\sqrt{4})^{-2}$. For each model, 10,000 MCMC samples are collected after a burn-in of 10,000 iterations. The estimands of interest in both sections are functions of quantiles; we will use $Q(s \mid \x)$ to denote the quantile function of $P(y \mid \x)$. 

\input{app1}
\input{app2}

\section{Conclusion}\label{sec:conclusion}
In this paper, we introduced a new nonparametric model, DR-BART, extending Bayesian Additive Regression Trees to the novel and challenging setting of density regression. This model has the appeal of being flexible yet easy to specify and understand, with few prior parameters to set. This distinguishes DR-BART from other nonparametric Bayesian methods for density regression, which often include collections of infinite dimensional regression parameters. Inference via MCMC is fast, with acceptable mixing in our examples obtained in a matter of minutes for datasets with thousands of observations. We showed DR-BART to empirically outperform a variety of other density regression methods in its ability to point estimate conditional densities and to express uncertainty about these estimates. Lastly, we extended previous work on posterior concentration results for BART and for density estimation via a latent variable model to prove concentration rates for DR-BART.

Having introduced the latent variable as a modeling device a natural question is whether it might represent some real structure in the scientific problem, like measurement error in $\x$, some combination of omitted variables, or a latent construct like ability or motivation. Any of these seem plausible in the applications presented here. In the educational testing application we treated the latent variables as independent. This is a useful ``saturated'' model for the joint distribution of test scores, but models for dependence across these variables or lower-dimensional representations based on shared latent variables may be more scientifically meaningful or efficient. Further modeling of latent variables within BART is a promising area for future work, even beyond the setting of density regression.

\bibliographystyle{unsrtnat}
% \bibliography{references}  %%% Uncomment this line and comment out the ``thebibliography'' section below to use the external .bib file (using bibtex) .


\section{Supplement}

\subsection{BART}
In this section, we review the basic details of Bayesian Additive Regression Trees (BART). We refer the reader to \cite{Chipman2010} for a full exposition.
\subsubsection{Model and Priors}
\citeauthor{Chipman2010} consider the regression model
\[y_i = f(\bx_i) + \eps_i, \eps_i \stackrel{iid}{\sim} N(0, \sigma ^ 2)\]
and place a BART prior on $f$. Such a prior represents the function as the sum of $m$ many piecewise constant regression trees. 
Each tree $T_h,\;1\leq h\leq m$ consists of a set of interior decision nodes (where decisions are generally of the form $x_j<c$ for some value $c$) and a set of $b_h$ terminal nodes. The terminal nodes have associated parameters $\M_h = (\mu_{h1},\mu_{h2},\dots \mu_{hb_h})'$.  For each tree there is a partition of the covariate space $\{\mathcal{A}_{h1},\dots,\mathcal{A}_{hb_h}\}$ with each element of the partition corresponding to a terminal node. A tree and its associated parameters define step functions:
\beq
\g(\x, T_h, \M_h) = \mu_{hb}\text{ if }\x\in\mathcal{A}_{hb} \text{ (for $1\leq b\leq b_h$)}.
\eeq

These functions are additively combined to obtain $f$:

\begin{equation}
f(\x) = \sum_{h=1}^m \g(\x, T_h, \M_h).
\end{equation}

In the spirit of boosting, each term in the sum is constrained by a strong prior to be a ``weak learner"; that is, the prior on $(T_h, \M_h)$ strongly favors small trees and leaf parameters that are near zero. Each tree independently follows the prior described by \cite{Chipman1998}, where the probability that a node at depth $d$ splits (is not terminal) is given by 
\begin{equation}
\alpha (1+d)^{-\beta},\;\;\alpha\in (0,1),\;\beta\in [0,\infty).
\end{equation}
 A variable to split on, and a cutpoint to split at, are then selected uniformly at random from the available splitting rules. We follow CGM throughout by taking $\alpha=0.95$ and $\beta=2$. Traditionally the prior on cutpoints for the $j^{th}$ variable is a discrete uniform distribution over a uniformly spaced grid or some collection of quantiles of $\{x_{ij}: 1\leq i\leq n\}$. Most implementations of tree-based models also require that the terminal nodes not be empty, or contain at least 5 observations.

To set the priors on $M$, CGM suggest scaling the data to lie in $\pm \, 0.5$ and assign the leaf parameters independent priors:

\begin{equation}
\mu_{hb}\sim N(0,\sigma^2_\mu)\;\;\textrm{ where }\sigma_\mu=1/(2k\sqrt{m}).
\end{equation} 
CGM recommend $1\leq k\leq 3$, with $k=2$ as a reasonable default choice. This prior shrinks $g_l(x)$ strongly toward zero, while ensuring that the induced prior for $f(x)$ is centered at zero and puts approximately 95\% of the prior mass within $\pm 0.5$. Larger values of $k$ imply increasing degrees of shrinkage. Performance is fairly insensitive to the number of trees $m$, as long as it is large enough. In practice, $m = 200$ is a common choice. 

\subsubsection{Posterior Sampling}

\cite{Chipman2010} provide full details of the MCMC algorithm used to fit BART. A key ingredient of the sampler is performing blocked updates for $(T_h\M_h\mid \{(T_l, M_l\}_{l\neq h}, -)$. This ``backfitting'' step utilizes the fact that the full conditional for $(T_h, \M_h)$ depends on $\{(T_l, M_l\}_{l\neq h}$ and $y$ only through the residuals
\beq
R_{hi} = \left( y_i - \sum_{l\neq h}^m \ghxi{l}{i}\right) \sim N(\ghxi{h}{i}, \sigma^2),
\eeq
 These residuals follow the single tree model studied by \cite{Chipman1998} with parameters $(T_h, \M_h, \sigma^2)$, so the \cite{Chipman1998} Metroplis-Hastings update can be embedded within the BART Gibbs sampler to sample $(T_h, \M_h)$ jointly from their full conditional.
 Since $\M_h$ has a normal prior, the integrated likelihood $p(R_h\mid T_h, \sigma^2)$
is available in closed form. Therefore $(T_h, M_h)$ can be updated in a block by sampling $T_h$ marginally over $\M_h$ with a Metropolis step and then sampling $\M_h$ from its full conditional given the new $T_h$. Complete details including proposal distributions are given in \cite{Chipman1998} (the results in this paper use a smaller set of propsal distributions, outlined in \cite{Pratola2014}). This blocked update enhances mixing and obviates the need for transdimensional MCMC algorithms due to the changing dimensionality of $\M_h$.

\subsection{Details on Heteroscedastic Regression with BART priors}\label{sec:app:hetero}
\subsubsection{Model and Priors}
Here, we give an overview of the \cite{murray2021log} extension of the BART prior to model variance functions. Begin with the loglinear model
\beq
%R_{hi} = \left( y_i - \sum_{l\neq h}^m \right) \sim N(\ghxi{h}{i}, \sigma^2),
v(x) = \sum_{h=1}^{m_v} g(\x, \vT_h, \vM_h)
\eeq
 where $\{(\vT_h$, $\vM_h)\}$ are trees and parameters for the variance functions. To collect posterior samples, a variant of the backfitting MCMC algorithm is possible. Suppose we are updating $(\vT_h,\vM_h)$. The first step is constructing the scaled residuals
%
%\begin{equation}
%r_i = \frac{(y_i - f(x_i))}{\exp(\sum_{h=2}^m \vg(x_i)/2)}.
%\end{equation}
\begin{equation}
R^{(v)}_{hi} = \frac{(Y_i - f(x_i))}{\exp[\sum_{l\neq h}^m g(\x, \vT_l, \vM_l)]}.
\end{equation}
Now $R^{(v)}_{hi}\sim N(0, \sigma^2\exp[g(\x, \vT_h, \vM_h)])$ 
so we have
\beq
p(R^{(v)}_h\mid \vT_h, \vM_h, \sigma^2) = 
 \prod^{b^{(v)}_h}_{l=1}%N(\mu_{hl}; 0, \sigma^2_\mu) 
	\prod_{i:\x_i\in \mathcal{A}_l} 
	\frac{1}{\sigma_0\exp[\mu^{(v)}_{hl}/2]}
	\phi\left(\frac{R^{(v)}_{ih}}{\sigma_0^2\exp[\mu^{(v)}_{hl}]}\right)
\eeq

As a normal prior for $\mu^{(v)}_{hb}$ is no longer conditionally conjugate, \citeauthor{murray2021log} introduces another prior which is symmetric about $0$ on the log scale and also admits closed form representations for the integrated likelihood and full conditional distribution. 
%different conditionally conjugate prior, 
This prior is specified as 
%A conditionally conjugate prior for $\vtau_{hb}\equiv \exp(\vM_{hb})$ that does satisfy these criteria is 
an equal probability mixture of gamma and inverse gamma distributions with the same parameters. Let $\vtau_{hb} := \exp(\mu_{hb}^{(v)})$. Then, the prior is:

\begin{equation}
p\left(\vtau_{hb}\right) %= 0.5p_{IG}(\vM_{hb}; a_vm, a_vm) + 0.5p_{G}(\vM_{hb}; a_vm, b_vm)
=\frac{1}{2}
\frac{b^a}{\Gamma(a)}\left(\vtau_{hb}\right)^{-a-1}\exp\left(-b/\vtau_{hb}\right)
+\frac{1}{2}\frac{b^a}{\Gamma(a)}\left(\vtau_{hb}\right)^{a-1}\exp\left(-b\vtau_{hb}\right)\label{eq:ptau}  
\end{equation}

Throughout this paper we take $a = b = a_0 m_v$. To motivate this choice, note that if $\vtau_{hb}$ has density \eqref{eq:ptau} with $a = b = a_0 m_v$ then $\log(\vtau_{hb})\eqd (-1)^{Z_h}\log(W_h)$ where $Z_h\iid Bern(0.5)$ and $W_h\iid Ga(a_0m_v, a_0m_v)$. The first two moments of $\log(\vtau_{hb}))$ are 0 and $(a_0m_v)^{-1}$ and for large $a_0m_v$ its distribution is nearly normal since the Gamma distributions are so concentrated near 1. Thus, marginally, $v(x)$ will be approximately distributed $N(0,1/a_0)$ and can be scaled to match the \emph{a priori} expected range of the log-variance process. Even for relatively small $a_0m_v$ this approximation is excellent; see Figure \ref{fig:hetprior}. The marginal prior for $\log(v(x))$ is nearly normal for $a_0m_v=5$ and is indistinguishable from the normal density when $a_0m_v$ is greater than 25. In practice $m_v$ will generally be at least 100. To set $a_0$, note that $a_0 = [\log(\sqrt{d})]^{-2}$ implies that $\Pr(\exp[v(\x)]\in (\sigma^2_0/d, d\sigma_0^2))\approx 0.95$. For example, taking $a_0 = 1.5$ makes $d$ about $5$. In practice, values of $d$ around 2-4 tend to work well; much larger and the risk of overfitting and degenerate bandwidths increases.

\begin{figure}
{ \includegraphics[width=1\linewidth]{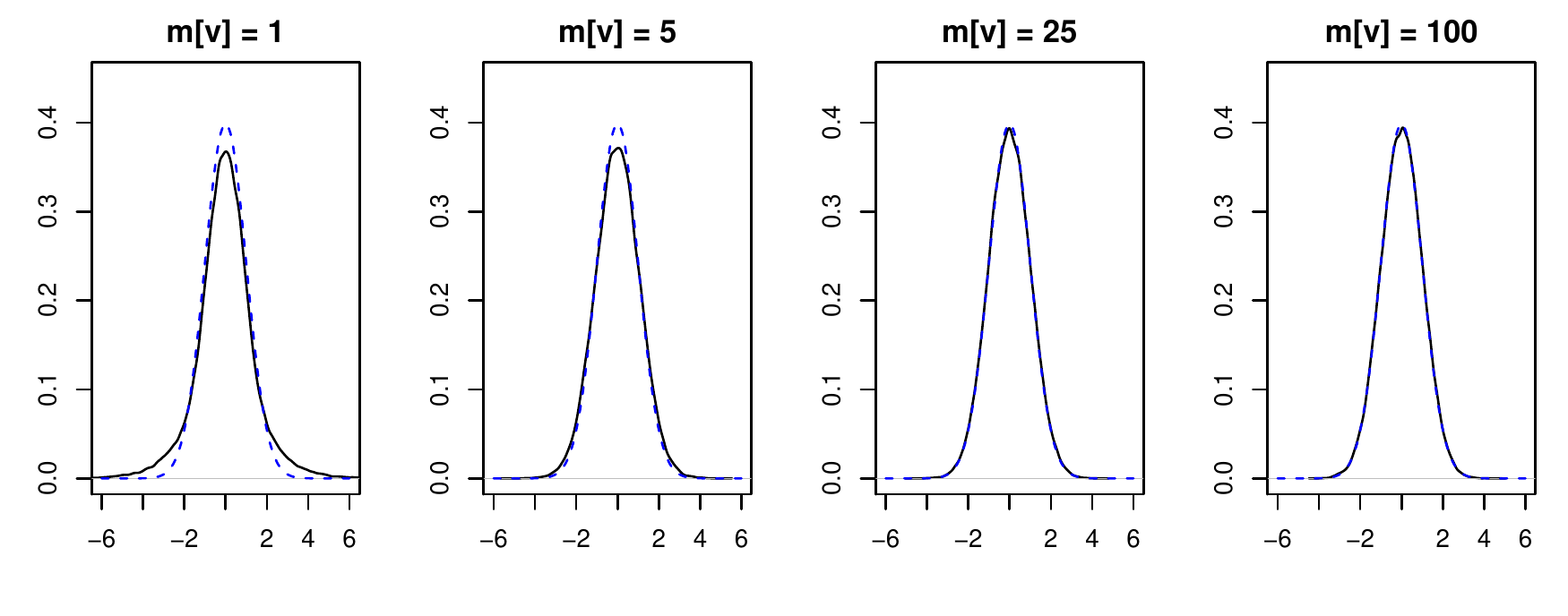} 
}
\caption{Marginal prior on $v(x)$ for $a_0=1$ and a range of $m_v$. The blue dashed line is the standard normal density.}
\label{fig:hetprior}
\end{figure} 

\subsubsection{Posterior Sampling}

Let $n_{hb} = \sum_{i=1}^n \ind{\x_i\in \mathcal{A}^{(v)}_{hb}}$, the number of observations in leaf $b$ of $\vT_h$, and let $r^2_{hb} = \sum_{i=1}^n [R^{(v)}_{hi}]^2\ind{\x_i\in \mathcal{A}^{(v)}_{hb}}$ be the sum of squared residuals in each leaf. Under \eqref{eq:ptau} the full conditional distribution for $\vtau_{hb}$ is another mixture distribution:
\beq
\begin{split}
p(\vtau_{hb}\mid R^{(v)}_h, \sigma^2)\propto
\frac{\Gamma(n_{hb}/2 + a)}{(b+r^2_{hb}/2)^{n_{hb}/2+a}}p_{IG}(\vtau_{hb}, n_{hb}/2 + a, b+r_{hb}^2/2)\\
+\frac{2K_{(n_{hb}/2-a)}(\sqrt{2br_{hb}^2})}
{(2b/r_{hb}^2)^{(a-n_{hb}/2)/2}}p_{GIG}(a-n_{hb}/2, b, r_{hb}^2/2)\label{eq:taufc}
\end{split}
\eeq
where $K_p(x)$ is the modified Bessel function of the second kind, $p_{IG}$ is the pdf of an inverse gamma random variable with rate $b$, and $p_{GIG}$ is the pdf of a generalized inverse Gaussian random variable:
\beq
p_{GIG}(\tau; \lambda, \psi, \chi) = \frac{(\psi/\chi)^{(\lambda/2)}}{2K_\lambda(\sqrt{\psi\chi})}\tau^{\lambda-1}\exp\left[{-\frac{1}{2}} \left( \psi\tau + \chi/\tau\right) \right]
\eeq
The integrated likelihood is 
\beq
P(R^{(v)}_h\mid \vT_h, \sigma^2) = 
 \prod^{b^{(v)}_h}_{l=1}%N(\mu_{hl}; 0, \sigma^2_\mu) 
	%\prod_{i:x_i\in \mathcal{A}_l} 
\frac{1}{2(2\pi)^{n_{hl}/2}}
\frac{b^a}{\Gamma(a)}
\left(\frac{\Gamma(n_{hl}/2 + a)}{(b+r^2_{hl}/2)^{n_{hl}/2+a}}
+\frac{2K_{(n_{hl}/2-a)}(\sqrt{2br^2_{hl}})}
{(2b/r_{hl}^2)^{(a-n_{hl}/2)/2}}\right).\label{eq:precil}
\eeq
Block updating  $(\vT_h, \vM_h)$ proceeds by substituting the integrated likelihood \eqref{eq:precil} into the Metropolis step for $\vT_h$ and drawing $\vM_h$ by sampling each $\tau_{hb_h}^{(v)}$ from \eqref{eq:precil}. Given $v$, updates for the rest of the parameters are straightforward although the integrated likelihood and full conditionals also have to be adjusted in the update for $(T_h,\M_h)$ to account for the heteroscedastic residuals. See \cite{Bleich2014} for details.

%%%%%%%%%%%%%%%%
\subsection{Proofs of Theorems}\label{sec:app_proofs}
\subsubsection*{Proof of Theorem \ref{thm:finite}}
Let $N_d$ be the number of nodes at depth $d$ of a tree sampled from this prior. There are $2^{d-1}$ pairs of possible nodes at depth $d$. Label these pairs by $1\leq i\leq 2^{d-1}$ and define $U_i^{(d)}$ as 1 if the $i^{th}$ pair of nodes exists and 0 otherwise, so that $N_d = \sum_{i=1}^{2^{d-1}}2U_i^{(d)}$. Now $U_i^{(d)}=1$ if and only if all its parents are splitting nodes, so $\Pr(U_i^{(d)}=1) = \E(U_i^{(d)}) = \prod_{s=0}^{d-1} \alpha(1+s)^{-\beta} = \alpha^d(d!)^{-\beta}$. Then $E(N_d) = \sum_{i=1}^{2^{d-1}}2\alpha^d(d!)^{-\beta} = (2\alpha)^d(d!)^{-\beta}$. Since $N_d$ is nonnegative, $\Pr(N_d\geq 1)\leq \E(N_d)$. Since $\sum_{d=0}^\infty\E(N_d) = \sum_{d=0}^\infty (2\alpha)^d(d!)^{-\beta}<\infty$, $\Pr(N_d\geq 1\ i.o.)=0$ so the tree is almost surely finite.

\subsubsection*{Proof of Theorem \ref{thm:posterior_concentration}}

As discussed in Section \ref{sec:theory}, we must find constants $C_1$, ..., $C_4$ and a sieve $\mathcal{G}_n$ that allow us to verify conditions G1, G2, and G3. We first prove the conditions using Lemmas \ref{lem:thick} -- \ref{lem:link}; in the next subsection we prove a series of Lemmas allowing us to establish Lemma $\ref{lem:link}$.
\begin{proof}[Proof of G1]
Let $\delta_n = (\log n / n)^{\beta/(2\beta+d_0)}$ and let $\epsilont_n = (\log n / n)^{\frac{\alpha}{\alpha+1} \times \frac{\beta}{2\beta + d_0}}$. 

\noindent Applying Lemma~\ref{lem:link} with $\epsilon = \epsilont_n$, for sufficiently large $n$ we have
  \begin{align}
    \label{eq:bound1}
    \begin{split}
    -\log\Pi\{\KL_{p_0}(C_{\KL} \epsilont_n)\}
    &\le
    -\log \Pi(S = S_0)
    - \log \Pi(\|f - f_0\|_\infty \le \delta_n \mid S = S_0)
    \\&- \log \Pi\{\sigma(\x,u) \equiv \sigma \text{ is constant},
    \sigma \in [\epsilont_n^{1/\alpha}, 2\epsilont_n^{1/\alpha}] \mid S = S_0\}.
  \end{split}
  \end{align}
  
  As argued by \citet{artofbart}, P1 implies $-\log\Pi(S = S_0) \lesssim d_0 \log(p + 1) \lesssim n \epsilon_n^2$. Next, by Lemma~\ref{lem:thick} we have $-\log \Pi(\|f - f_0\|_\infty \le \delta_n \mid S = S_0) \lesssim n\delta_n^2 \lesssim n \epsilon_n^2$. 
  Finally, note that with $\sigma_n = \epsilont_n^{1/\alpha}$, by P4 and standard properties of the inverse gamma distribution we can bound the final term of \eqref{eq:bound1} by
  \begin{align*}
    &-\log\Pi\{e^\xi \in [1.2 \sigma_n, 1.8 \sigma_n],
    \sup_{hl} |\mu_{hl}^{(v)}| \le 0.1/m_v,
    \text{all trees have depth 0}
    \mid S = S_0\}
    \\
    &\quad\lesssim
    -\log\Pi\{e^\xi \in [1.2 \sigma_n, 1.8 \sigma_n]\}
    % \\&
    \lesssim    
    \sigma_n^{-1}
    \lesssim n \epsilon_n^2.
  \end{align*}
  Combining these facts, we get $-\log \Pi\{\KL_{p_0}(C_{\KL} \epsilont_n)\} \lesssim n \epsilon_n^2$, and because $\epsilont_n \le \epsilon_n$ this implies that G1 holds for some choice of $C_1$ and $C_2$.
\end{proof}
\begin{proof}[Proof of G2 and G3]
  Let $\kappa$ be a large constant to be determined later and define $\mathcal F_f$ as in Lemma~\ref{lem:sieve} with the choices $d = \kappa n\epsilon_n^2 / \log(p+1)$, $K = \lfloor \kappa n \epsilon_n^2 / \log n \rfloor$, $U^2 = \kappa n\epsilon_n^2$, and $\epsilon = \epsilon_n / (n\epsilon_n^2)$. Let $\tilde{v}$ denote the tree-based contribution to the variance function, without the intercept: $\tilde{v} := v(\bx, u) - \xi$. Then, we similarly define $\mathcal F_{\tilde v}$ but with the choice $\epsilon = \epsilon_n^2$ and $U = V$. Lastly, define $\ell = 1/(\kappa n\epsilon_n^2)$ and $u = e^{\kappa n\epsilon_n^2}$. We take the set $\mathcal G$ to be given by
  \begin{align*}
    \mathcal G = \{p_{f,\sigma} : f \in \mathcal F_f, \tilde{v} \in \mathcal F_{\tilde{v}}, \xi \in [\log \ell, \log u]\}.
  \end{align*}
  Let $\widehat F_f$ and $\widehat F_{\tilde{v}}$ denote $\epsilon_n / (n\epsilon_n^2)$ and $\epsilon_n^2$ nets of $\mathcal F_f$ and $\mathcal F_{\tilde{v}}$ respectively and let $\widehat R$ denote an $\epsilon_n^2$-net of $[\log \ell, \log u]$. Given a $p_{f_1, \sigma_1} \in \mathcal G$ we can find $f_2 \in \widehat F_f$, $\tilde{v}_2 \in \widehat F_{\tilde{v}}$, and $\xi_2 \in \widehat R$ such that, by Lemma \ref{lem:link}:
    \begin{align*}
    h(p_{f_1, \sigma_1}, p_{f_2, \sigma_2})
    \le
    C_h \left\{\sqrt{2\epsilon_n^2} + e^{m_vV}\frac{\epsilon_n (\kappa n\epsilon_n^2)}{n\epsilon_n^2} \right\}
    \le
    C_\kappa \epsilon_n,
  \end{align*}

  for a global constant $C_h$ and a constant $C_{\kappa}$ depending only on the constant $\kappa$ and the prior. Hence
  \begin{align*}
    \log N(\mathcal G, C_\kappa \epsilon_n, h)
    \le
    \log N(\mathcal F_f, \epsilon_n / (n\epsilon_n^2), \|\cdot\|_\infty)
    + \log N(\mathcal F_{\tilde{v}}, \epsilon_n^2, \|\cdot\|_\infty)
    + \log \frac{\log u + \log \ell^{-1}}{\epsilon_n^2}.
  \end{align*}
  By Lemma~\ref{lem:sieve}, each term is easily verified to be bounded by a constant multiple of $\kappa n \epsilon_n^2$ so that $\log N(\mathcal G, C_\kappa \epsilon_n, h) \lesssim \kappa n \epsilon_n^2$.
  
 Next we compute the complementary probability bound. First, by the union bound:
  \begin{align}
    \label{eq:target}
    \begin{split}
    \Pi(p \notin \mathcal G)
    &\le
    \Pi(\xi < \log \ell) + \Pi(\xi > \log u) + \Pi(D > d)
    \\&\quad +
    \Pi(f \notin \mathcal F_f \mid D \le d) + \Pi(v \notin \mathcal F_{\tilde{v}} \mid D \le d).
    \end{split}
  \end{align}
  Using standard properties of the inverse gamma distribution we have
  \begin{align*}
    \log \Pi(\xi < \log \ell)
    &\lesssim
    -\ell^{-1}
    =
    -\kappa n \epsilon_n^2, \\
    \log \Pi(\xi > \log u)
    &\lesssim
      -\log u  = -\kappa n \epsilon_n^2.
  \end{align*}
  By P1 we also have $\log \Pi(D > d) \lesssim -d \log(p + 1) = -\kappa n \epsilon_n^2$. Next, if $D \le d$ but $f \notin \mathcal F_f$ then either (i) $K_m > K$ for some $m$ where $K_m$ is the number of leaf nodes in tree $m$ or (ii) $K_m \le K$ for all $m$, but $\sup_{h\ell} |\mu_{h\ell}| > U$. Hence, by the tail properties of Gaussian random variables,
  \begin{align*}
    \Pi(f \notin \mathcal F_f \mid D \le d)
    &\le
    \Pi(K_m > K \text{ for some $m$}) + \Pi(\sup_{h\ell} |\mu_{h\ell}| > U \mid K_m \le K \text{ for all $m$})
    \\&\le
    m \Pi(K_1 > K) + Km e^{-U^2/(2 \sigma^2_\mu)}.
  \end{align*}
  
  As noted in the proof of Theorem 2 of \citet{artofbart} we have $\log \Pi(K_1 > K) \lesssim -K \log K \lesssim -\kappa n\epsilon_n^2$ while $-U^2/(2\sigma^2_\mu) \lesssim -\kappa n\epsilon_n^2$ by our choice of $U$. Hence $\log \Pi(f \notin \mathcal F_f \mid D \le d) \lesssim -\kappa n \epsilon_n^2$ as well, provided that $n$ is sufficiently large. Finally, the only way for $\tilde{v} \notin \mathcal F_{\tilde{v}}$ to occur when $D \le d$ is for at least one tree to have more than $K$ leaf nodes, and we have already seen that the log of this probability is bounded by a constant multiple of $-\kappa n \epsilon_n^2$. Putting all of these facts together and bounding each term of \eqref{eq:target} by the maximum, we have
  \begin{align*}
    \Pi(p \notin \mathcal G)
    \le
    \exp\{-C_{\mathcal G} \kappa n \epsilon_n^2\}
  \end{align*}
  for some small constant $C_{\mathcal G}$ depending only on the prior. By taking $\kappa$ larger than $(C_2 + 4) / C_{\mathcal G}$ we have $\Pi(p \notin \mathcal G) \le \exp\{-(C_2 + 4) n \epsilon_n^2\}$ and $\log N(\mathcal G, \bar \epsilon_n, h) \le C_5 n \bar \epsilon_n^2$ for some constant $C_5$ and $\bar \epsilon_n$ a constant multiple of $\epsilon_n$, as desired.
\end{proof}

\subsubsection*{Proof of Supporting Lemmas}

As before, we use $a \lesssim b$ to mean there is a positive constant $C$ which can be computed from the prior and $p_0$, independent of $n$ and $p$, such that $a \le Cb$. Unless otherwise stated, the constant is assumed to be universal in the sense that if we write $f(x;\xi) \lesssim g(x)$ then we have $f(x;\xi) \le C g(x)$ for all $(x, \xi)$ (unless $\xi$ is part of the prior or a function of $p_0$).

Throughout this section, $\phi_\sigma * p_0(y\mid\x)$ denotes the convolution $\int \phi_\sigma(y - z) \, p_0(z\mid\x) \ dz$ where $\phi_\sigma(z) = (2\pi\sigma^2)^{-1/2} e^{-z^2/(2\sigma^2)}$. When $\sigma(\x,u)$ is a constant we note that $\phi_\sigma * p_0(y\mid\x) = p_{f_0, \sigma}$.

\begin{lemma}
  \label{lem:bounding-below}
  Suppose $p_0$ satisfies Condition F. Then for $\sigma < 1$, we have
  \begin{align*}
    p_0(y\mid \x)
    \lesssim
    \phi_{\sigma} * p_0(y\mid\x)
  \end{align*}
\end{lemma}

\begin{proof}
  By Condition F, $p_0(y\mid\x)$ is uniformly bounded and bounded away from 0 on $[0,1]$. First, if $y \notin [0,1]$ then we have $\phi_\sigma * p_0(y \mid \x) > 0 = p_0(y \mid \x)$ so that $\phi_\sigma * p_0(y \mid \x) > C p_0(y \mid \x)$ for any positive constant $C$; hence we can assume without loss of generality that $y \in [0,1]$. Let $M = \inf_{y,\x} p_0(y\mid\x)$ and write
  \begin{align*}
    \phi_\sigma * p_0(y\mid\x)
    &=
    % \int \phi_\sigma(y - z) \, p_0(z \mid \x) \ dz
    % \\&\ge
    \int_0^1  \phi_\sigma(y - z) \, p_0(z \mid \x) \ dz
    \\&\ge
    M \int_0^1 \phi_\sigma(y - z) \ dz
    \\&=
    M \int_{y-1}^y \phi_\sigma(u) \ du.
    \\&\ge
    M \int_0^1 \phi_\sigma(u) \ du
    \ge M / 3
    \ge \frac{M}{3 \sup_{y,\x} p_0(y\mid\x)} p_0(y\mid\x),
  \end{align*}
  where the second inequality follows from the fact $U(y; \sigma) = \Phi\{y / \sigma\} - \Phi\{(y - 1) / \sigma\}$ is minimized at $y = 1$ for $y\in[0,1]$ and all $\sigma$, and the third inequality follows from the fact that $U(1; \sigma) > 1/3$ for all $\sigma < 1$. Hence $\phi_\sigma * p_0(y\mid\x) \gtrsim p_0(y \mid \x)$ for $\sigma < 1$.
\end{proof}

\begin{lemma}
  \label{lem:ghoshbound}
  If $p_0$ satisfies Condition F and $\sigma(\x,u)$ is constant with $\sigma < 1$ then
  \begin{align*}
    \log \|p_0 / p_{f, \sigma}\|_\infty
    \le
    C_{\inf} + \frac{\|f-f_0\|^2_\infty}{\sigma^2},
  \end{align*}
  for some $C_{\inf}$.
\end{lemma}

\begin{proof}
  The proof is essentially the same as the proof of Lemma 6.2 of \citet{gp}. We write
  \begin{align*}
     p_{f,\sigma}(y \mid \x)
    &= \frac{1}{\sqrt{2\pi} \sigma} \int_0^1 \exp\left\{ -\frac{(y - f(\x,u))^2}{2\sigma^2} \right\}
    \\&\ge \sqrt 2\exp\{-\|f - f_0\|_\infty^2/\sigma^2\}
           \int_0^1 \frac{1}{\sqrt{2\pi \sigma^2/2}} \exp\left\{ -\frac{(y-f_0(\x,u))^2}{\sigma^2} \right\} 
    \\&= \sqrt{2}\exp\{-\|f - f_0\|_\infty^2/\sigma^2\} \phi_{\sigma/\sqrt 2} * p_0(y\mid\x).
  \end{align*}
  Therefore, using the fact that $\phi_{\sigma/\sqrt 2} * p_0(y\mid\x) \gtrsim p_0(y\mid\x)$ we have
  \begin{align*}
    \log \frac{p_0(y\mid\x)}{p_{f,\sigma}(y\mid\x)}
    \le
    C_{\inf} + \frac{\|f - f_0\|_\infty^2}{\sigma^2}
  \end{align*}
  for some constant $C_{\inf}$.
\end{proof}

\begin{lemma}
  \label{lem:kruj}
  Let $p_0$ satisfy Condition F. Then $\|p_0 - \phi_\sigma * p_0\|_\infty \lesssim \sigma^\alpha$ for sufficiently small $\sigma$.
\end{lemma}

\begin{proof}
 Suppose $\alpha \le 1$. Let $L = \sup_{(y,\x) \ne (y',\x')} \frac{|p_0(y\mid\x) - p_0(y'\mid\x')|}{\|(y-y', \x-\x')\|^\alpha}$, which is finite by Condition F. Then:
  \begin{align*}
    |p_0(y\mid\x) - \phi_\sigma * p_0(y\mid\x)|
    &= |p_0(y\mid\x) - \int \phi_\sigma(u)p_0(y-u\mid\x)\,du| \\
    &= |\int p_0(y\mid\x)\phi_\sigma(u)\,du- \phi_\sigma(u)p_0(y-u\mid\x)\,du|\\
    &\le \int |p_0(y\mid\x)\phi_\sigma(u)-p_0(y-u\mid\x)\phi_\sigma(u)|\,du\\
    &\le
    L \int |u|^\alpha \, \phi_\sigma(u) \ du \lesssim
    \sigma^\alpha.
  \end{align*}
  
  For $1 < \alpha \le 2$, let $\dot{p}_0(y\mid\x) = \frac{d}{dy} p_0(y\mid\x)$ and let $L = \sup_{(y,\x) \ne (y',\x')} \frac{|\dot{p}_0(y\mid\x) - \dot{p}_0(y'\mid\x')|}{\|(y,\x)-(y',\x')\|^{\alpha - 1}}$, which is finite by Condition F. \\Then:
  \begin{align*}
    p_0(y - u\mid\x) = p_0(y\mid\x) - u\dot{p}_0(y'\mid\x) + u\dot{p}_0(y) - u\dot{p}_0(y\mid\x)
  \end{align*}
  for some $y'$ between  $y$ and $y - u$, depending on $(y,\x,u)$. Integrating with respect to $\phi_\sigma(u)$, the term $u \dot{p}_0(y\mid\x)$ drops; then taking the absolute value we have:
  \begin{align*}
    \bigg|\int (p_0(y - u\mid\x) - p_0(y\mid\x))\phi_\sigma(u)\,du\bigg| &= \bigg|\int u(\dot{p}_0(y\mid\x) - \dot{p}(y'\mid\x))\phi_\sigma(u)\,du\bigg|
  \end{align*}  
  The left hand side becomes $|p_0(y\mid\x) - \phi_\sigma * p_0(y\mid\x)|$ and the right hand side is upper bounded by:
  \begin{align*}
      \int |u(\dot{p}_0(y) - \dot{p}_0(y'))\phi_\sigma(u)|\,du \leq L\int |u|^\alpha \phi_\sigma(u) \,du \lesssim \sigma^\alpha
  \end{align*}
  yielding the stated result.
\end{proof}

\begin{lemma}
  \label{lem:kl}
  Suppose Condition F holds. Then for sufficiently small $\sigma$ we have
  \begin{align*}
    \int p_0(y\mid\x) \log \frac{p_0(y\mid\x)}{\phi_\sigma * p_0(y\mid\x)}
    \ dy \ F_\bX(d\x)
    \lesssim
    \sigma^{2\alpha}.
  \end{align*}
\end{lemma}

\begin{proof}
  For fixed $\x$ we can bound the Kullback-Leibler divergence in terms of the chi-squared divergence as
  \begin{align*}
    \int p_0(y\mid\x) \log \frac{p_0(y\mid\x)}{\phi_\sigma * p_0(y\mid\x)}
    \ dy \ F_\bX(d\x)
    \le
    \int \frac{(p_0(y\mid\x) - \phi_\sigma * p_0(y\mid\x))^2}{\phi_\sigma * p_0(y\mid\x)} \ dy \ F_{\bX}(d\x).
  \end{align*}
  By Lemma~\ref{lem:bounding-below} and Condition F, we have that the denominator is bounded away from 0 for small enough $\sigma$. Lemma~\ref{lem:kruj} further tells us that the numerator is $O(\sigma^{2\alpha})$, giving the result.
\end{proof}

\subsubsection*{Proof of Lemma~\ref{lem:link}}
We first establish the second statement. Let $h^2_{\x}(p, q) = \int \{p(y \mid \x)^{1/2} - q(y \mid \x)^{1/2}\} ^ 2 \ dy$ denote the usual Hellinger distance between $p(\cdot \mid \x)$ and $q(\cdot \mid \x)$. We begin by noting that, by Cauchy-Schwarz,
\begin{align*}
  p_{f_1,\sigma_1}(y\mid\x) \, p_{f_2,\sigma_2}(y\mid\x)
  \ge
  \left\{\int_0^1 \sqrt{\phi_{\sigma_1}(y - f_1(\x,u)) \phi_{\sigma_2}(y-f_2(\x,u))}\right\} ^ 2.
\end{align*}
By Fubini's Theorem and the usual formula for the affinity between normal distributions we have
\begin{align}
  \label{eq:holder}
  \begin{split}
  h^2_{\x}(p_{f_1,\sigma_1}, p_{f_2,\sigma_2})
  &\le
  1 - \int_0^1 \int \sqrt{\phi_{\sigma_1}(y - f_1(\x,u)) \phi_{\sigma_2}(y-f_2(\x,u))} \ dy \ du
  \\&=
  \int_0^1 1 - \sqrt{\frac{2\sigma_1(\x,u)\sigma_2(\x,u)}{\sigma_1(\x,u)^2 + \sigma_2(\x,u)^2} }\exp\left\{-\frac{(f_1(\x,u) - f_2(\x,u))^2}{4(\sigma_1(\x,u)^2 + \sigma_2(\x,u)^2)}\right\}  \ du
  \end{split}
\end{align}
By the triangle inequality we have $h_\x^2(p_{f_1, \sigma_1}, p_{f_2, \sigma_2}) \le 2\{h_\x^2(p_{f_1, \sigma_1}, p_{f_1, \sigma_2}) + h^2_{\x}(p_{f_1, \sigma_2} p_{f_2, \sigma_2})\}$. Now, note from the inequality $1 - e^{-|x|} \le |x|$ that for any $\sigma_1 , \sigma_2$ we have
\begin{align*}
  1 - \sqrt{\frac{2\sigma_1 \sigma_2}{ \sigma_1^2 + \sigma_2^2}}
  \le
  1 - \sqrt{\frac{\sigma_1^2\wedge \sigma_2^2}{\sigma_1^2 \vee \sigma_2^2}}
  =
  1 - \exp\{-|\log\sigma_1 - \log \sigma_2|\}
  \le
  |\log \sigma_1 - \log \sigma_2|.
\end{align*}
Hence using \eqref{eq:holder} with $f_1$ in place of $f_2$ gives
\begin{align*}
  h_\x^2(p_{f_1, \sigma_1}, p_{f_1, \sigma_2})
  \le
  \int_0^1 |\log \sigma_1(\x,u) - \log \sigma_2(\x,u)| \ du
  \le
  \|\log \sigma_1 - \log \sigma_2\|_\infty.
\end{align*}
Next, we use \eqref{eq:holder} with $\sigma_2$ in place of $\sigma_1$ to get
\begin{align*}
  h_\x^2(p_{f_1,\sigma_2}, p_{f_2,\sigma_2})
  &\le
  \int_0^1 1 - \exp\left\{- \frac{(f_1(\x,u) - f_2(\x,u))^2}{4\sigma_2(\x,u)^2}
  \right\} \ du
  \\&\le \int_0^1 \frac{(f_1(\x,u) - f_2(\x,u))^2}{4\sigma_2(\x,u)^2} \ du
  \\&\le \frac{\|f_1 - f_2\|^2_\infty}{4\inf_{u}\sigma_2(\x,u)^2}.
\end{align*}
Putting these together we get
\begin{align*}
  h_{\x}^2(p_{f_1,\sigma_1},p_{f_2,\sigma_2})
  \le
  2\left\{
  \|\log \sigma_1 - \log \sigma_2\|_\infty
  +
  \frac{\|f_1 - f_2\|^2_\infty}{4 \inf_{u} \sigma_1(\x,u)^2 \wedge \sigma_2(\x,u)^2}
  \right\}.
\end{align*}
Integrating with respect to $F_{\bX}(d\x)$ and taking the square root gives the result.

To prove the first bound, fix $\x$ and suppose $p_{f,\sigma}$ is such that $\sigma(\x,u)$ is constant, $\sigma \in (\epsilon^{1/\alpha}, 2\epsilon^{1/\alpha})$, and $\|f - f_0\|_\infty \le \epsilon^{1+1/\alpha}$. Note that $\phi_\sigma * p_0 = p_{f_0,\sigma}$. For any $p_{f,\sigma}$, applying the triangle inequality we have
\begin{align*}
  h_{\x}^2(p_0, p_{f,\sigma})
  \le
  2\{h_{\x}^2(p_0, p_{f_0,\sigma}) + h^2_{\x}(p_{f_0,\sigma},p_{f,\sigma})\}.
\end{align*}
Using the fact that $h_{\x}^2(p, q) \le \int p \log(p / q) \ dy$, the proof of Lemma~\ref{lem:kl} gives $h_{\x}^2(p_0, p_{f_0,\sigma}) \lesssim \sigma^{2\alpha} \lesssim \epsilon^2$ for sufficiently small $\epsilon$. Next, we have
\begin{align*}
  h^2_{\x}(p_{f_0,\sigma}, p_{f,\sigma})
  \le
  1 - \exp\left\{ - \frac{\|f - f_0\|^2_\infty}{8\sigma^2} \right\}
  \lesssim
  \frac{\|f - f_0\|^2_\infty}{\sigma^2}
  \lesssim
  \frac{\epsilon^{2 + 2/\alpha}}{\epsilon^{2/\alpha}}
  \lesssim
  \epsilon^2.
\end{align*}

Combining these facts and integrating with respect to $F_\bX(d\x)$ gives $h^2(p_0, p_{f\sigma}) \lesssim \epsilon^2$. An application of Lemma 8 of \citet{ghosal_vdv} and Lemma~\ref{lem:ghoshbound} gives $K(p_0 \| p_{f,\sigma}) \lesssim \epsilon^2$ and $V(p_0 \| p_{f,\sigma}) \lesssim \epsilon^2$ as well; let $C_{\KL}$ denote the constant which makes $\KL(p_0 \| p_{f,\sigma}) \vee V(p_0 \| p_{f,\sigma}) \le C^2_{\KL} \epsilon^2$. Hence $p_{f,\sigma} \in \KL(C_{\KL} \epsilon)$, proving the result.

\subsection{Simulation Details}
\subsubsection{Implementation}
Code for conducting these simulations can be found at: \url{https://github.com/vittorioorlandi/DR-BART}.

Implementation details are as below: 
\begin{itemize}
    \item \textbf{DR-BART} DR-BART was implemented in R using the \texttt{Rcpp} package. 
    \item \textbf{DR-GP} We implemented DR-GP via the \texttt{rstan} package on CRAN for interfacing with Stan \citep{stan}. We followed \cite{Kundu2014}, used a squared exponential kernel for the Gaussian Process, and introduced covariates appropriately. 
    \item \textbf{SBART-DS} The code for SBART-DS was graciously provided by \citeauthor{modbart} 
    \item \textbf{PSBPM} This method was run via the implementation on the author's GitHub at: \url{https://github.com/david-dunson/probit-stick-breaking}. Prior parameters were as suggested in \cite{pbss}.
    \item \textbf{DPMM} The code for DPMM was taken from version 1.17 of the archived CRAN package \texttt{DPpackage}, which can be found here: \url{https://cran.r-project.org/src/contrib/Archive/DPpackage/}. Prior parameters were left at their defaults. 
\end{itemize}
All methods were run for 25,000 iterations of burn-in, after which 25,000 posterior samples were saved. 

\subsubsection{Predicted Coverage and Credible Band Width}
Here, we provide additional information on 1. the predictive coverage and 2. the credible band width of each method across various simulation settings. To compute the predictive coverage, an additional $n = 1000$ test points were generated for each run of a simulation. The predicted coverage for a run is the proportion of the test points that were contained in the 95\% HDR intervals; the reported values are averages across all runs of a simulation. The reported widths are average widths -- across runs of a simulation -- of 95\% posterior credible bands within the 95\% HDR region of the true density. That is, we evaluate credible band width in a high density region of the data. 
The predictive coverage below shows that all methods expect for PSBPM consistently have good predictive coverage. The information on credible bands is useful in conjunction with the coverage results in the main text, as it helps show which of the methods that undercover tend to do so because they are unreasonably confident in their estimates (e.g. DR-GP and PSBPM) versus those that simply do not capture the shape of the density well enough (e.g. SBART-DS). 

\begin{table}[h]
\centering
\begin{tabular}{|l||l|l|l|l|l|}
\hline
$x$   & DR-BART & DPMM  & SBART-DS & DR-GP & PSBPM \\ \hline
0.1 & 0.95    & 0.95 & 0.93 & 0.96 & 0.63 \\ \hline
0.5 & 0.95    & 0.95 & 0.96 & 0.96 & 0.75  \\ \hline
0.8 & 0.99    & 0.97  & 0.97 & 0.99 & 0.00  \\ \hline
\end{tabular}
\caption{Mean predictive coverage of 95\% predictive intervals, across runs of Simulation 1.}
\label{tab:sim1_pred_cvg}
\end{table}

\begin{table}[h]
\centering
\begin{tabular}{|l||l|l|l|l|l|}
\hline
$x$   & DR-BART & DPMM  & SBART-DS & DR-GP & PSBPM \\ \hline
0.1 & 0.12    & 0.13 & 0.07 & 0.08 & 0.07 \\ \hline
0.5 & 0.26   & 0.20  & 0.17 & 0.10 & 0.12  \\ \hline
0.8 & 0.67   & 0.33  & 0.27 & 0.20 & 0.00  \\ \hline
\end{tabular}
\caption{Mean 95\% credible band width within the true 95\% HDR region, averaged across runs of Simulation 1.}
\label{tab:sim1_band_width}
\end{table}

\begin{table}[h]
\centering
\begin{tabular}{|l||l|l|l|l|l|}
\hline
$x$   & DR-BART & DPMM  & SBART-DS & DR-GP & PSBPM \\ \hline
0.1 & 0.31    & 0.1 & 0.07 & 0.05 & 0.08 \\ \hline
0.5 & 0.53    & 0.1  & 0.17 & 0.07 & 0.12  \\ \hline
0.8 & 1.55    & 0.3  & 0.27 & 0.20 & 0.00  \\ \hline
\end{tabular}
\caption{Mean 95\% credible band width within the true 95\% HDR region, averaged across runs of Simulation 2.}
\label{tab:sim2_band_width}
\end{table}

\begin{table}[h]
\centering
\begin{tabular}{|l||l|l|l|l|l|}
\hline
$x$   & DR-BART & DPMM  & SBART-DS & DR-GP & PSBPM \\ \hline
0.1 & 0.92    & 0.94 & 0.93 & 0.90 & 0.64 \\ \hline
0.5 & 0.94    & 0.97 & 0.95 & 0.97 & 0.73  \\ \hline
0.8 & 1.00    & 0.99  & 0.97 & 1.00 & 0.00  \\ \hline
\end{tabular}
\caption{Mean predictive coverage of 95\% predictive intervals, across runs of Simulation 2.}
\label{tab:sim2_pred_cvg}
\end{table}

\begin{table}[h]
\centering
\begin{tabular}{|l||l|l|l|l|l|}
\hline
$x$   & DR-BART & DPMM  & SBART-DS & DR-GP & PSBPM \\ \hline
0.1 & 0.95    & 0.95 & 0.94 & 0.96 & 0.61 \\ \hline
0.5 & 0.95    & 0.97 & 0.96 & 0.97 & 0.78 \\ \hline
0.8 & 1.00    & 0.96  & 0.97 & 0.99 & 0.00  \\ \hline
\end{tabular}
\caption{Mean predictive coverage of 95\% predictive intervals, across runs of Simulation 3.}
\label{tab:sim3_pred_cvg}
\end{table}

\begin{table}[h]
\centering
\begin{tabular}{|l||l|l|l|l|l|}
\hline
$x$   & DR-BART & DPMM  & SBART-DS & DR-GP & PSBPM \\ \hline
0.1 & 0.11    & 0.13 & 0.07 & 0.08 & 0.07 \\ \hline
0.5 & 0.33    & 0.34  & 0.24 & 0.16 & 0.17  \\ \hline
0.8 & 0.63    & 0.33  & 0.25 & 0.20 & 0.00  \\ \hline
\end{tabular}
\caption{Mean 95\% credible band width within the true 95\% HDR region, averaged across runs of Simulation 3.}
\label{tab:sim3_band_width}
\end{table}

\subsection{Applications}
Code for performing these analyses can be found at: \url{https://github.com/vittorioorlandi/DR-BART}. The data for the application on returns to education can be found \href{https://www.econometricsociety.org/content/supplement-quantile-regression-under-misspecification-application-us-wage-structure-data}{here}. The data for the student growth application can be found in the \texttt{SGPdata} R package.

\end{document}

%% file: app1.tex
\subsection{Calculating Student Growth}\label{sec:app_edu}
We consider the estimation of student growth from a series of test scores, using anonymized mathematics test scores provided in the SGP R package \citep{Betebenner2011SGP}. \cite{Castellano2013} gives an overview of current mean and quantile regression approaches to this problem, as well as some consideration to the substantive problem of measuring student growth. Statistically, the problem reduces to estimating a series of conditional densities $p(y_t\mid y_{t-1}, y_{t-2},\dots,y_0)$. These distributions are heteroscedastic and tend to become skewed as they approach the extremes.  The current state of the art is the quantile regression methodology implemented in the SGP package and detailed in \cite{Betebenner2009}. The quantile regression models are specified as
\beq
Q_{y_t}(s\mid y_{t-1},\dots y_0) = \sum_{r=0}^{t-1} \eta_r^{(s)}(y_r).\label{eq:SGP}
\eeq
where each function $\eta_r^{(s)}(\cdot)$ is given a B-spline basis expansion. This model is fit at each of the 99 percentiles to approximate the conditional distribution. This procedure does not yield valid estimates for percentiles of the conditional distribution in general, due to quantile crossing (though this can be corrected post-hoc \citep{Chernozhukov2010}). Assessing uncertainty in this framework is challenging as well, and analysis generally relies solely on point estimates. Bayesian density regression addresses both problems simultaneously.

We fit independent DR-BART-LH models to scores from grades 4-7, each conditional on all the previous scores as well as the grade 3 scores. For simplicity we took a subsample of 3,000 students with complete data from grades 3-7. An interesting feature of this data is the interactions between previous test scores on the predictive distributions. One mechanism for this is regression to the mean: large jumps in test scores are unlikely to be sustained. Figure \ref{fig:interaction} gives an example, displaying posterior predictive densities for grade 5 scores when fixing grade 3 and 4 scores at all combinations of their marginal quartiles. In the absence of interactions the shift between the pairs of solid and dashed densities would be equal. Figure \ref{fig:interactionq} further suggests the presence of an interaction effect. When $y_4=455$, $y_3$ has a smaller effect on the quantiles of $p(y_5\mid y_4, y_3)$ than when $y_4 = 538$. Note that the additive model in \eqref{eq:SGP} cannot capture such an interaction and requires the curves to be equal.

\begin{figure}[h!]
\begin{center}
{\centering \includegraphics[height=.45\linewidth,width=.7\linewidth]{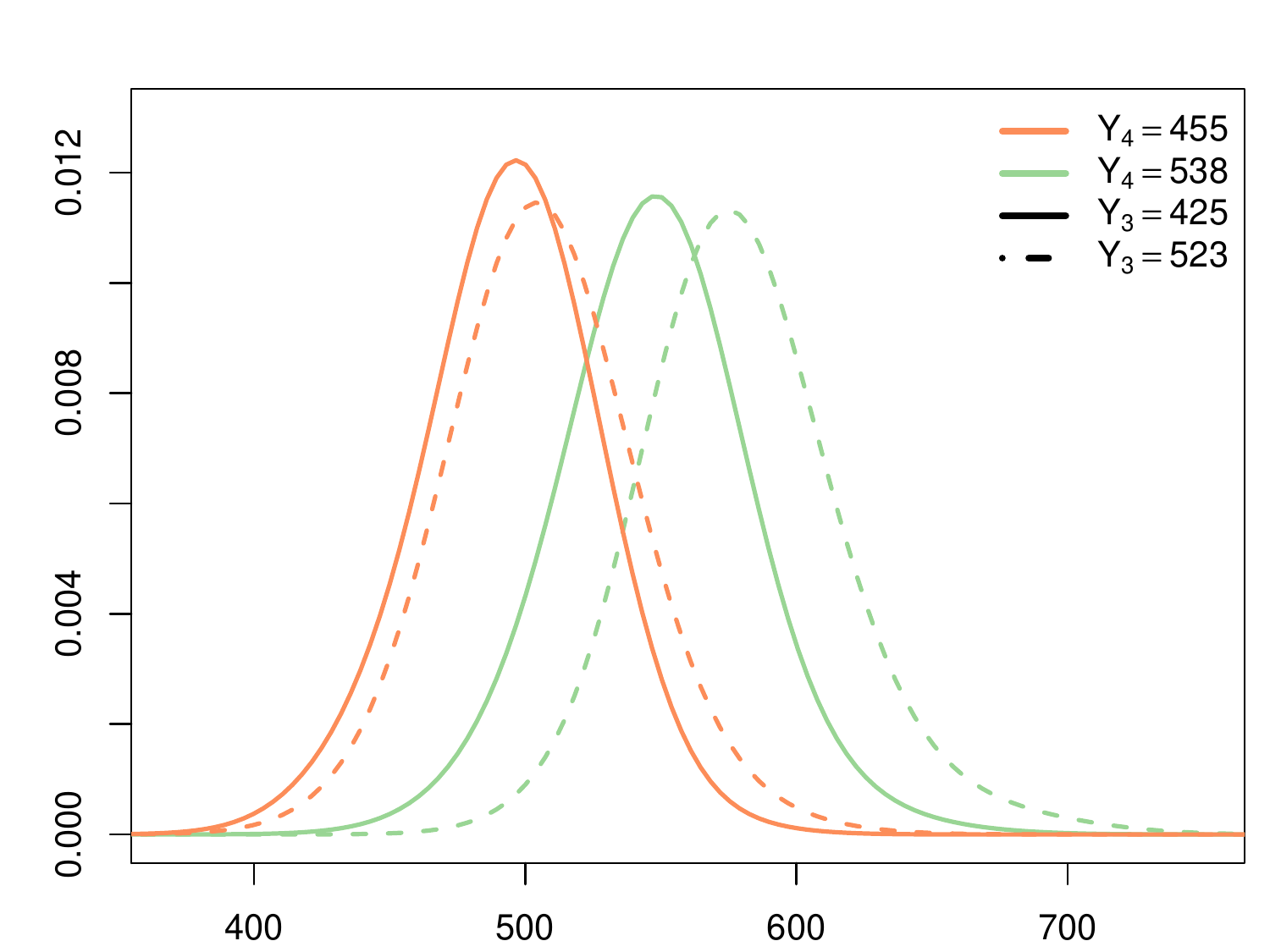} 
}
\end{center}
\caption{Four posterior predictive distributions for grade 5 scores obtained by setting grade 3 and 4 scores at each combination of their marginal quartiles. The effect of the grade 3 score on the predictive distribution clearly depends on the value of the grade 4 score, an important interaction that DR-BART is well-equipped to capture.}
\label{fig:interaction}
\end{figure} 

\begin{figure}[h!]
\begin{center}
{\centering \includegraphics[height=.45\linewidth,width=.95\linewidth]{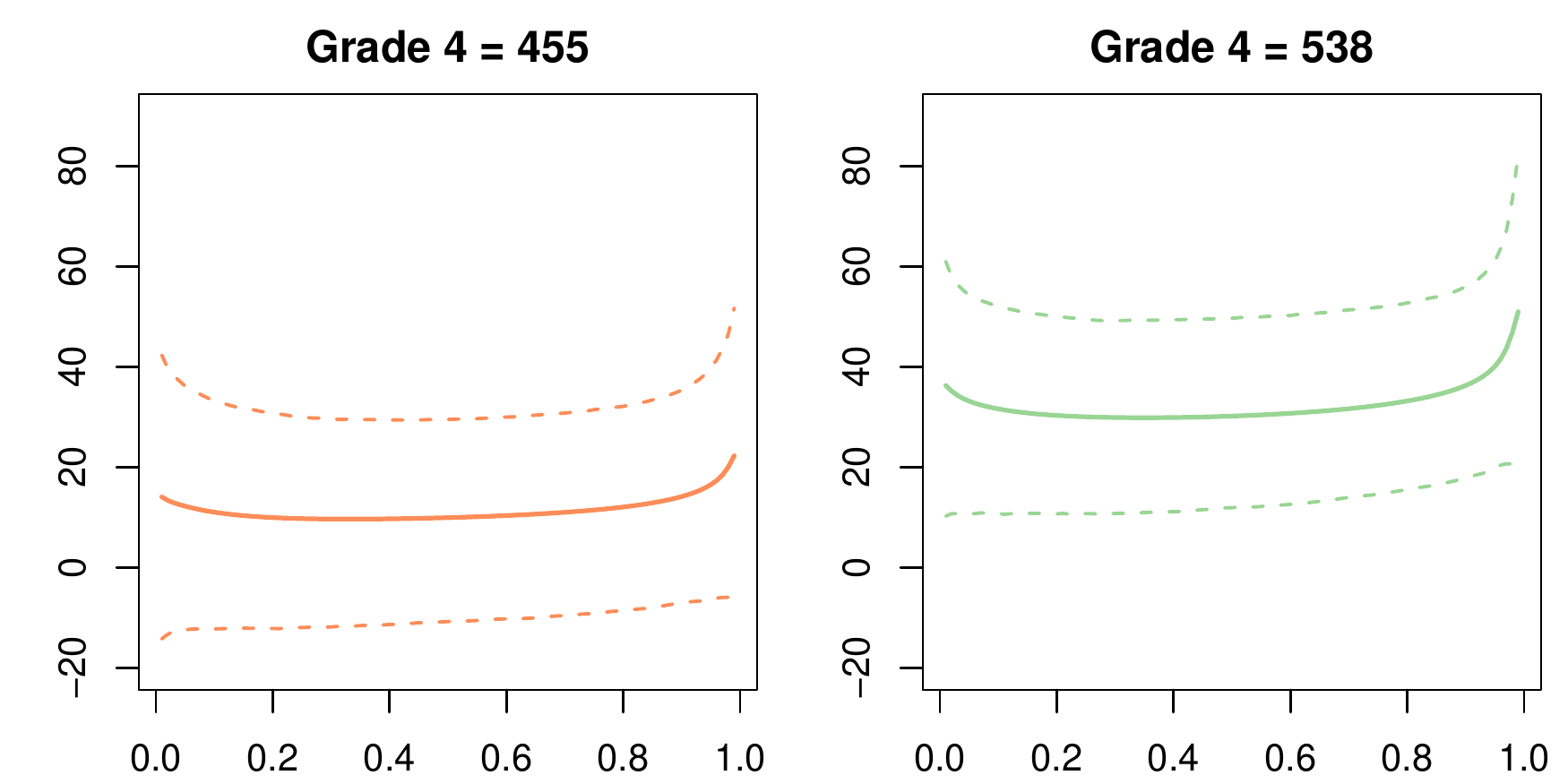} 
}
\end{center}
\caption{Quantile effects $Q_{y_5}(s\mid y_4, y_3=523) - Q_{y_5}(s\mid y_4, y_3=425)$ for $y_4 = 455,538$ along with 90\% credible intervals. Under \cite{Betebenner2009}'s additive model these two curves are forced to be equal. This reveals a likely interaction between grade 3 and grade 4 scores: the posterior mean of the 538 curve is nearly outside the 455 credible interval.}
\label{fig:interactionq}
\end{figure} 

A key objective in student growth modeling is computing growth quantile targets. A growth quantile is a quantile of the predictive distribution for the current test score given the score history. Intuitively, it provides a measure of how well a student performed on the test relative to academic peers (the hypothetical population of students with identical test history). A growth quantile target is the level of consecutive quantile growth that would be required to achieve pre-established achievement targets. For these data, there are four achievement levels used to define achievement targets:
% this particular dataset achievement targets are three thresholds per grade defining four achievement levels: 
Unsatisfactory, Partially Proficient, Proficient, or Advanced. Growth quantile targets answer questions like ``What level of sustained growth is necessary for a student with grade 3 test score $y_3$ to be Proficient by grade 7?''. This is different than simulating score trajectories and computing the probability that a student reaches a target. Growth quantile targets are intended to promote a ``what will it take'' attitude over a fatalistic ``where will s/he be'' attitude \citep{Betebenner2011}.

The current methodology for computing growth quantile targets uses a series of point estimates, ignoring uncertainty in the quantile curve estimates. But this uncertainty is generally not negligible, particularly for students with extreme test scores. To illustrate, we computed posterior samples of growth quantile curves for a hypothetical third grade student who scored at the cusp of Unsatisfactory/Partially Proficient on the third grade math test. Figure \ref{fig:sgpunsat} plots this student's  growth quantile curves and 90\% credible intervals for grades 4-7. 
%and another who scored at the boundary of Partially Proficient/Proficient (Figure \ref{fig:sgppprof}). 
%These curves show the posterior mean and 90\% credible intervals of the student's test score if they were to acheive growth at the $p^th$ quantile for each intervening year.
It is easy to read off growth quantile targets from these charts: For example, the bottom right panel of Figure \ref{fig:sgpunsat} shows that if the student sustains $60^{th}$ percentile growth in grades 5,6 and 7, there is a 95\% probability of remaining partially proficient in grade 7. Having a 50\% chance of reaching proficiency in grade 7 would require sustained growth at the $80^{th}$ percentile, which is probably unattainable without some new intervention. Similar curves can be computed starting with any grade and any score history. %Alternative growth targets are also possible to compute.  It would be easy to simulate growth curves that specify a \emph{range} of quantiles rather than a single value, to (for example) compute projections and attainment probabilities assuming the student sustains growth .

\begin{figure}[h!]
\begin{center}
{\centering \includegraphics[width=.8\linewidth]{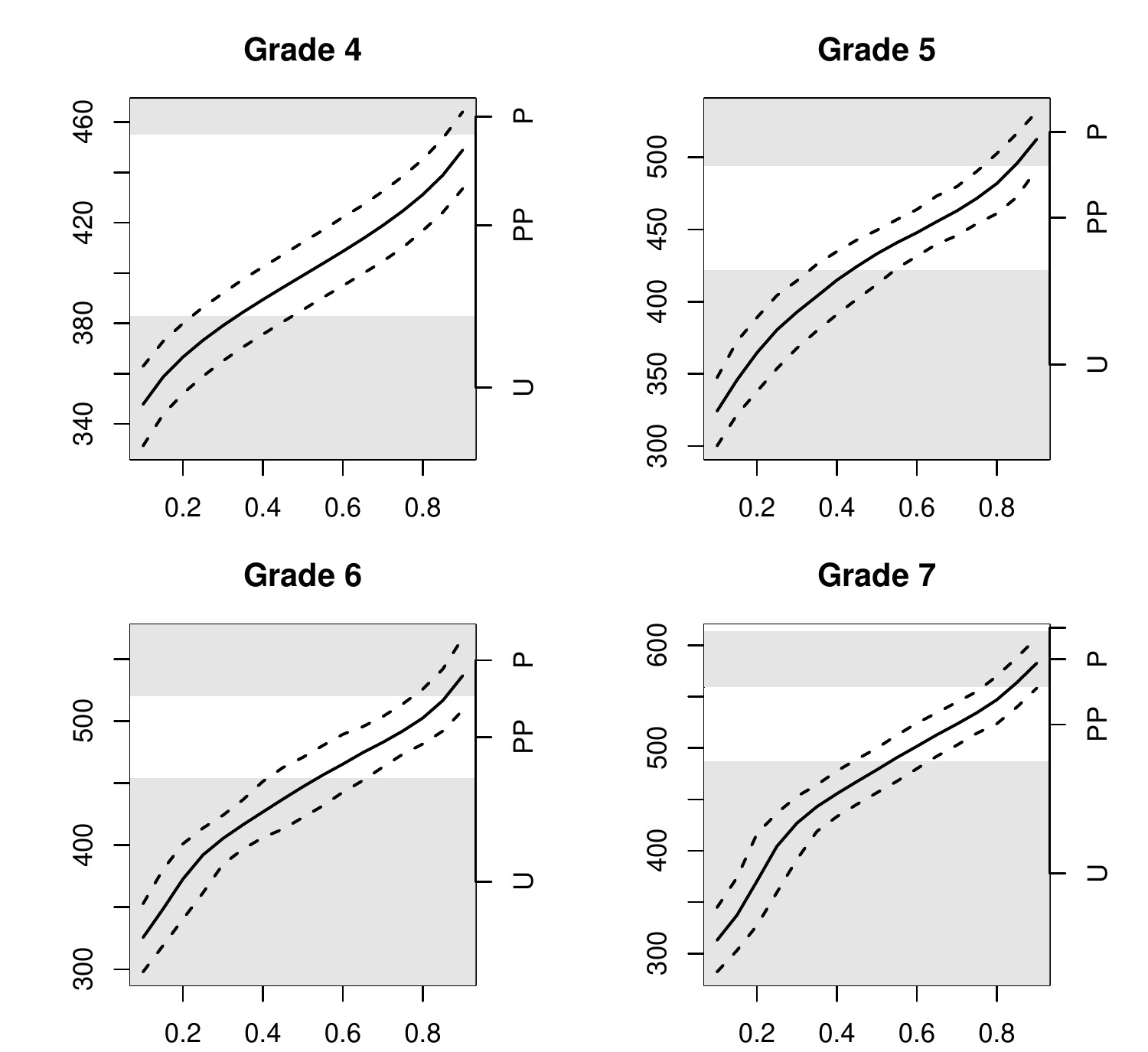} 
}
\end{center}
\caption{Projected scores as a function of quantile growth for a student at the cusp of Unsatisfactory/Partial Proficiency in grade three. The solid line is the posterior mean and dashed lines are pointwise 90\% credible intervals. Shaded regions correspond to achievement thresholds: Unsatisfactory (U), Partially Proficient (PP), Proficient (P).}
\label{fig:sgpunsat}
\end{figure}

% \begin{figure}
% \begin{center}
% {\centering \includegraphics[width=.9\linewidth]{figures/sgpproj-419} 
% }
% \end{center}
% \caption{Projected scores as a function of quantile growth for a student at the cusp of Unsatisfactory/Partial Proficiency in grade three (scale score 419). The solid line is the posterior mean, and dashed lines are pointwise 90\% credible intervals. Shaded regions correspond to acheivement thresholds: Unsatisfactory (U), Partially Proficient (PP), Proficient (P), Advanced (A).}
% \label{fig:sgppprof}
% \end{figure} 

%% file: app2.tex
In particular, we analyze a subset of public use microdata from the U.S. Census, originally compiled by \cite{Angrist2006} to estimate returns to education across the income distribution. For the 1980, 1990 and 2000 samples they extracted all U.S. born white and black men aged 40-49 with positive annual earnings and positive hours worked in the year prior to the Census. Records with imputed values were excluded, and wages were adjusted to 1989 dollars. \cite{Angrist2006} and its supplementary materials contain details of how the data were obtained and cleaned. The response is log monthly wages, and the covariates include years of education, experience (defined as $age-education-12$) and race (white or black). The study objective was to estimate returns to education as a function of quantile index. 
Taking $Q$ as the quantile function of wages, the return to education at quantile index $s$ is
\beq
100\times \frac{Q(s\mid \x_2) - Q(s\mid \x_1)}{Q(s\mid \x_1)},\label{eq:pmchange}
\eeq
the predicted percentage change in monthly wage from modifying education in $\x_1$ to yield $\x_2$. While \cite{Angrist2006} fit a linear quantile regression to estimate the effect of each additional year of education, we fit DR-BART-LH to each census sample separately, subsampling 5000 units from each, and study the difference between 12 and 16 years of education. In general the effect of a single additional year of schooling is probably heterogenous. For example, the difference in earnings between 11 and 12 years of schooling should be larger than the difference between 10 and 11, because 12 years of education in the U.S. typically indicates that the respondent completed high school. BART readily accommodates non-smooth features in the regression function, while usual econometric analyses assume that regression functions are linear or quadratic in education (or experience).
%Their model assumes the regression functions are a quadratic function of . This is unlikely; on average, the difference in earnings between 11 and 12 years of schooling should be larger than the difference between 10 and 11, because 12 years of education typically indicates that the respondent completed high school. BART readily accommodates such near-discontinuities in the regression funciton.

Figure \ref{fig:earn} shows the return to education for a 45 year old white male, comparing 12 to 16 years of schooling. In 1990, returns were highest at low and high quantiles, whereas in 2000 the returns are actually increasing as a function of the quantile index. 
% In other words, in 2000 the returns are actually lower for low income men than for middle or high income men. 
Other covariate vectors show similar patterns, although the exact estimates vary (and are somewhat unstable for black men due to small sample size). \cite{Angrist2006} found a similar pattern in the returns at the population level and verified it using additional data from the Current Population Survey, so this seems to be a robust finding.

\begin{figure}[h!]
\begin{center}
{\centering \includegraphics[height=.45\linewidth,width=.475\linewidth]{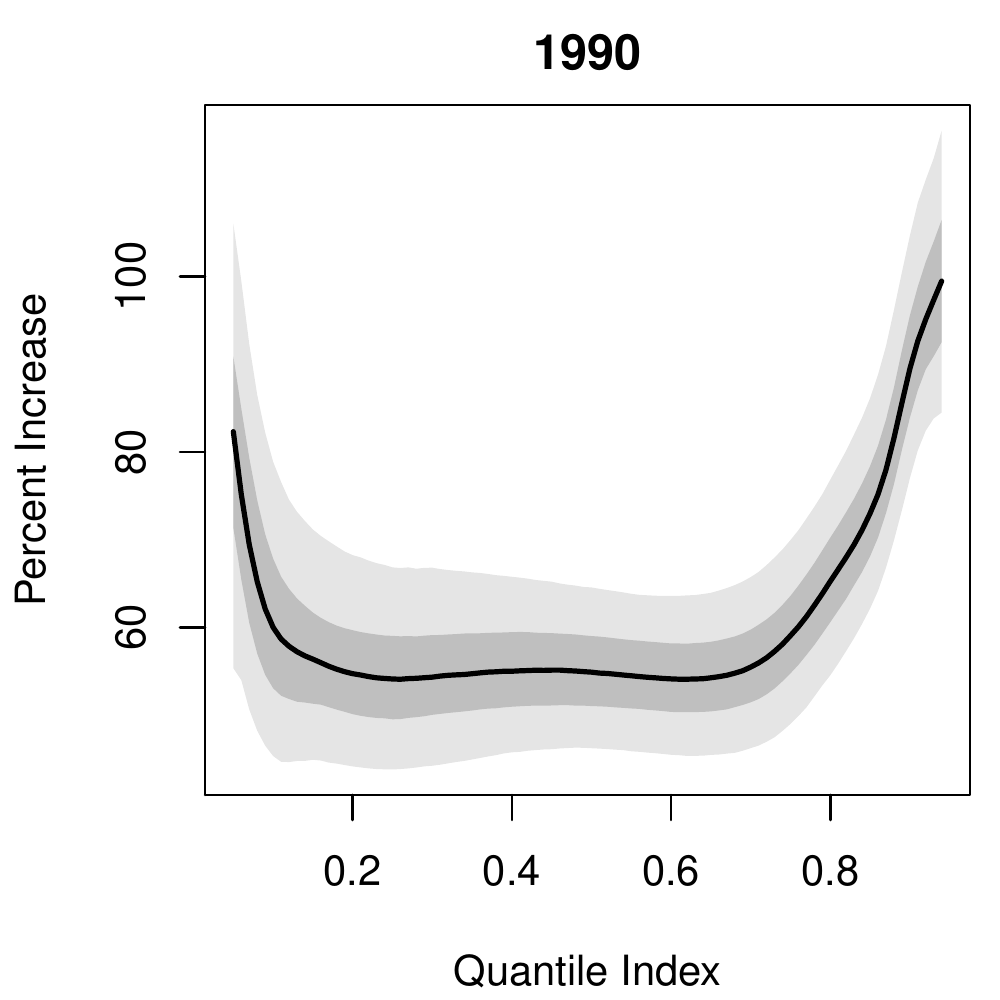} 
\includegraphics[height=.45\linewidth,width=.475\linewidth]{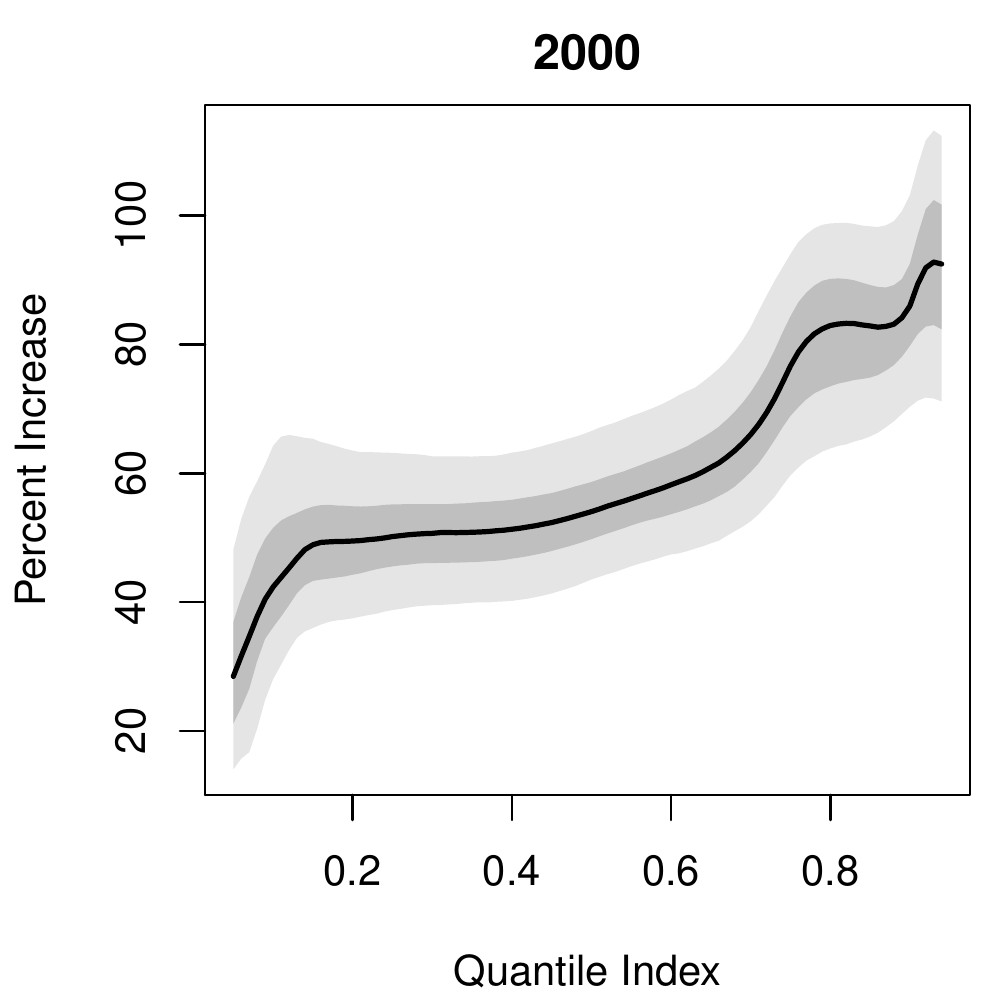} 
}
\end{center}
\caption{Returns to education (16 years versus 12 years) for a 45 year old white man.}
\label{fig:earn}
\end{figure}